\documentclass[conference]{IEEEtran}

\usepackage{subcaption}
\usepackage{graphicx}
\usepackage{amsmath,amssymb}
\usepackage{algorithm}
\usepackage{algpseudocode}
\usepackage{cite}
\usepackage[
    colorlinks=true,
    linkcolor=blue,
    citecolor=blue,
    urlcolor=blue,
    breaklinks=true
]{hyperref}

\usepackage{makecell}
\usepackage{amsthm}
\newtheorem{theorem}{Theorem}
\newtheorem{definition}{Definition}

\makeatletter
\renewcommand\section{\@startsection{section}{1}{\z@}%
  {1ex}{0.3ex}%
  {\normalfont\normalsize\centering\scshape}}
\renewcommand\subsection{\@startsection{subsection}{2}{\z@}%
  {1ex}{0.3ex}%
  {\normalfont\normalsize\itshape}}
\renewcommand\subsubsection{\@startsection{subsubsection}{3}{\z@}%
  {1ex}{0.3ex}%
  {\normalfont\normalsize\itshape}}
\makeatother

\begin{document}

\title{A Risk-Aware UAV-Edge Service Framework for Wildfire Monitoring and Emergency Response}

\author{
{Yulun Huang, Zhiyu Wang, and Rajkumar Buyya}
\\
\\ \textit{Quantum Cloud Computing and Distributed Systems (qCLOUDS) Laboratory},
\\ \textit{School of Computing and Information Systems},
\\ \textit{The University of Melbourne}, Melbourne, Australia
\\ \{yulun.huang@student.unimelb.edu.au, zhiywang1@unimelb.edu.au, rbuyya@unimelb.edu.au\}
}

\maketitle

\begin{abstract}
Wildfire monitoring demands timely data collection and processing for early detection and rapid response. UAV-assisted edge computing is a promising approach, but jointly minimizing end-to-end service response time while satisfying energy, revisit time, and capacity constraints remains challenging. We propose an integrated framework that co-optimizes UAV route planning, fleet sizing, and edge service provisioning for wildfire monitoring. The framework combines fire-history-weighted clustering to prioritize high-risk areas, Quality of Service (QoS)-aware edge assignment balancing proximity and computational load, 2-opt route optimization with adaptive fleet sizing, and a dynamic emergency rerouting mechanism. The key insight is that these subproblems are interdependent: clustering decisions simultaneously shape patrol efficiency and edge workloads, while capacity constraints feed back into feasible configurations. Experiments show that the proposed framework reduces average response time by 70.6--84.2\%, energy consumption by 73.8--88.4\%, and fleet size by 26.7--42.1\% compared to GA, PSO, and greedy baselines. The emergency mechanism responds within 233 seconds, well under the 300-second deadline, with negligible impact on normal operations.
\end{abstract}


\section{INTRODUCTION}

Wildfires pose severe threats to ecosystems, property, and human lives, with increasing frequency and intensity driven by climate change~\cite{abatzoglou2025climate}. Recent catastrophic events highlight the scale of this challenge: Australia's 2019--2020 Black Summer bushfires burned over 24 million hectares, revealing critical gaps in existing monitoring infrastructure~\cite{sharples2016natural}, while the January 2025 Los Angeles wildfires destroyed over 18,000 structures and caused an estimated \$250--275 billion in damages, underscoring the urgent need for early detection in wildland-urban interface areas~\cite{LA2025wildfires}. Effective wildfire management hinges on early detection and rapid response, both of which demand timely collection and processing of environmental data from distributed ground sensors~\cite{li2025pdenet}. However, traditional stationary monitoring infrastructure struggles in remote forest areas due to limited network coverage and prohibitive deployment costs~\cite{liu2025joint, wang2024deep}. Unmanned Aerial Vehicles (UAVs) equipped with communication modules offer a compelling alternative, providing flexible and mobile data collection capabilities over large monitoring areas~\cite{zhang2025uav}.

From a service computing perspective, this monitoring pipeline can be abstracted as a Quality of Service (QoS)-constrained IoT provisioning problem, which we term Wildfire Monitoring as a Service (WMaaS): ground sensors generate data collection requests, UAVs act as mobile gateways, and edge nodes execute fire detection tasks~\cite{mu2025edge, wang2024tf}. Designing such a system, however, involves tightly coupled optimization challenges. UAVs must follow patrol routes under strict energy budgets~\cite{oubbati2025uav}, while route planning must minimize revisit intervals to ensure timely data collection and balance energy consumption across the fleet~\cite{shen2022energy}. Edge nodes must be strategically assigned to clusters considering both geographic proximity and computational load to meet QoS requirements~\cite{xie2024delay, wang2025reinfog}. Furthermore, the system must adapt dynamically when sensors detect potential fire events, triggering UAV rerouting for emergency data collection without disrupting ongoing operations~\cite{du2025survey}.

Existing approaches typically address these challenges in isolation. UAV route planning methods~\cite{yan2024cooperative, yu2022novel, zhou2018mobile} optimize patrol distances but neglect edge computing constraints. Joint UAV-edge optimization frameworks~\cite{Sun2024INFOCOM, He2024INFOCOM, Sun2024TMC} integrate trajectory planning with computational offloading but target generic IoT scenarios without domain-specific adaptation such as risk-aware prioritization. UAV-based wildfire detection systems~\cite{Bushnaq2021IoTJ, zhang2023mmfnet, wang2024rfwnet} focus on fire identification algorithms but do not jointly optimize data collection routing, resource allocation, or fleet sizing. A holistic framework that co-optimizes these interdependent decisions (UAV routing, fleet sizing, edge assignment, and emergency adaptation), while satisfying all operational constraints, remains an open problem.

We propose an integrated framework for UAV-assisted wildfire monitoring with edge computing. The framework employs fire-history-weighted K-means clustering to partition sensors into balanced groups while prioritizing high-risk areas. A QoS-aware edge assignment strategy considers both proximity and computational capacity when mapping clusters to edge nodes. Within each cluster, 2-opt local search minimizes patrol distances, and an adaptive fleet sizing procedure iteratively determines the minimum number of UAVs needed to satisfy energy and revisit time constraints. For emergency scenarios, a dynamic rerouting mechanism dispatches the nearest available UAV to collect urgent data, with seamless return to normal patrol upon completion.

The main contributions are as follows:
\begin{itemize}
\item We formulate UAV-assisted wildfire monitoring as a joint optimization problem over route planning, fleet sizing, and edge provisioning, subject to energy, revisit time, and capacity constraints. Unlike prior work that addresses these in isolation, our formulation captures their interdependence: clustering decisions simultaneously determine patrol efficiency and edge workload distribution.

\item We design a fire-history-weighted clustering algorithm that biases cluster formation toward areas with high fire risk. This is non-trivial because naively weighting sensors can create highly imbalanced clusters that violate energy constraints; our approach maintains balance through weighted centroid updates and edge-aware initialization.

\item We develop a two-phase edge assignment strategy combined with 2-opt route optimization that jointly minimizes end-to-end response time. The two-phase design is necessary because direct-communication sensors and UAV-mediated sensors impose fundamentally different load patterns on edge nodes.

\item We propose a dynamic emergency mechanism that dispatches the nearest available UAV to collect urgent data and seamlessly resumes normal patrol upon completion. The three-step protocol (dispatch, load-aware delivery, route resumption) targets each component of emergency response time while containing disruption to ongoing monitoring operations.
\end{itemize}

\section{RELATED WORK}

\subsection{UAV Route Planning}

UAV path planning is commonly formulated as TSP or VRP variants~\cite{Clarke1964, Toth2002} and solved using metaheuristics such as GA~\cite{yan2024cooperative}, PSO~\cite{yu2022novel}, and greedy heuristics~\cite{zhou2018mobile}. Recent work has incorporated energy constraints: Raj and Simmhan~\cite{Raj2023CCGrid} proposed a platform for personalized UAV fleets using edge and cloud resources, and Sun et al.~\cite{Sun2024INFOCOM} presented a two-time-scale joint optimization for UAV-assisted MEC. However, these methods assume fixed data collection patterns and do not account for application-specific requirements such as risk-aware area prioritization or dynamic rerouting for emergency events.

\subsection{Edge Computing and Task Offloading}

Edge computing for IoT workloads has been widely studied in the context of task offloading and load balancing. Wang et al.~\cite{Wang2024INFOCOM} addressed computation offloading in heterogeneous edge networks, and Kim et al.~\cite{kim2024distributed} proposed distributed scheduling for edge clusters. UAV-enabled MEC frameworks~\cite{He2024INFOCOM, Sun2024TMC} jointly optimize UAV trajectories and offloading decisions, but target generic IoT scenarios. These approaches assume stationary or predictable service topologies and do not address the challenges introduced by mobile UAV gateways: time-varying connectivity, intermittent data delivery, and the need to balance edge load across clusters whose workloads depend on routing decisions.

\subsection{UAV-Assisted Wildfire Monitoring}

The integration of UAVs with sensing infrastructure for wildfire detection has received growing attention. Bushnaq et al.~\cite{Bushnaq2021IoTJ} proposed a UAV-IoT network where ground sensors accumulate data and wait for patrolling UAVs to relay it to base stations. Other works have advanced fire detection through AI models~\cite{zhang2023mmfnet, wang2024rfwnet}, but focus on the detection component in isolation without jointly optimizing data collection routing, edge resource allocation, or fleet sizing.

\subsection{Positioning of This Work}

Table~\ref{tab:related_comparison} summarizes the comparison across six dimensions. Three gaps emerge from existing literature. First, UAV routing methods~\cite{yan2024cooperative, yu2022novel, zhou2018mobile} do not consider how routing decisions affect downstream edge workloads. Second, UAV-edge frameworks~\cite{Sun2024INFOCOM, He2024INFOCOM, Sun2024TMC} lack domain-specific adaptation, such as fire-risk-aware prioritization and emergency rerouting. Third, wildfire-focused systems~\cite{Bushnaq2021IoTJ, zhang2023mmfnet, wang2024rfwnet} treat detection and data collection as separate concerns without end-to-end QoS optimization.

\begin{table}[!t]
\centering
\caption{Comparison with Related Work}
\label{tab:related_comparison}
\scriptsize
\renewcommand{\arraystretch}{1.05}
\setlength{\tabcolsep}{2pt}
\begin{tabular}{lcccccc}
\hline
\textbf{Work} & \textbf{\makecell{UAV\\Route}} & \textbf{\makecell{Edge\\Service}} & \textbf{\makecell{Wildfire\\Domain}} & \textbf{\makecell{Risk-\\Aware}} & \textbf{\makecell{Emergency\\Adapt.}} & \textbf{\makecell{Fleet\\Sizing}} \\
\hline
Yan et al.~\cite{yan2024cooperative} & \checkmark & & & & & \checkmark \\
Yu et al.~\cite{yu2022novel} & \checkmark & & & & & \\
Zhou et al.~\cite{zhou2018mobile} & \checkmark & & & & & \\
Wang et al.~\cite{Wang2024INFOCOM} & & \checkmark & & & & \\
Kim et al.~\cite{kim2024distributed} & & \checkmark & & & & \\
Raj \& Simmhan~\cite{Raj2023CCGrid} & \checkmark & \checkmark & & & & \\
Sun et al.~\cite{Sun2024INFOCOM} & \checkmark & \checkmark & & & & \\
He et al.~\cite{He2024INFOCOM} & \checkmark & \checkmark & & & & \\
Sun et al.~\cite{Sun2024TMC} & \checkmark & \checkmark & & & & \\
Bushnaq et al.~\cite{Bushnaq2021IoTJ} & \checkmark & & \checkmark & & & \\
Zhang et al.~\cite{zhang2023mmfnet} & & & \checkmark & & & \\
Wang et al.~\cite{wang2024rfwnet} & & & \checkmark & & & \\
\textbf{Ours} & \checkmark & \checkmark & \checkmark & \checkmark & \checkmark & \checkmark \\
\hline
\end{tabular}
\end{table}

Our framework addresses these gaps jointly. A key insight driving the design is that the subproblems are interdependent: clustering decisions simultaneously determine patrol route efficiency and edge workload distribution, while edge capacity constraints in turn limit feasible cluster configurations. We therefore co-optimize clustering, edge assignment, routing, and fleet sizing within a single iterative procedure.

We adopt classical combinatorial optimization rather than DRL-based approaches~\cite{He2024INFOCOM, Sun2024TMC} for three reasons. First, our problem is a one-shot offline planning task where complete information (sensor locations, fire history, edge capacities) is available, eliminating the need for online learning. Second, deterministic constraint satisfaction is essential, as energy and deadline violations in safety-critical wildfire operations are unacceptable. Third, deployment in new regions requires no model retraining, enabling immediate operability.

\section{MODELING AND PROBLEM FORMULATION}

\subsection{Service Architecture}

The proposed system comprises three layers as shown in Fig.~\ref{fig:system}. At the bottom, stationary ground sensors and visual monitoring devices generate data collection requests. In the middle, UAVs equipped with cameras and communication modules serve as mobile gateways, collecting data from sensors within communication range and relaying it to edge nodes. At the top, edge servers execute fire detection algorithms and coordinate emergency response through alert broadcasting to fire management agencies. Long-term data is archived via cloud storage.

\begin{figure}[!htb]
\centering
\includegraphics[width=0.95\linewidth]{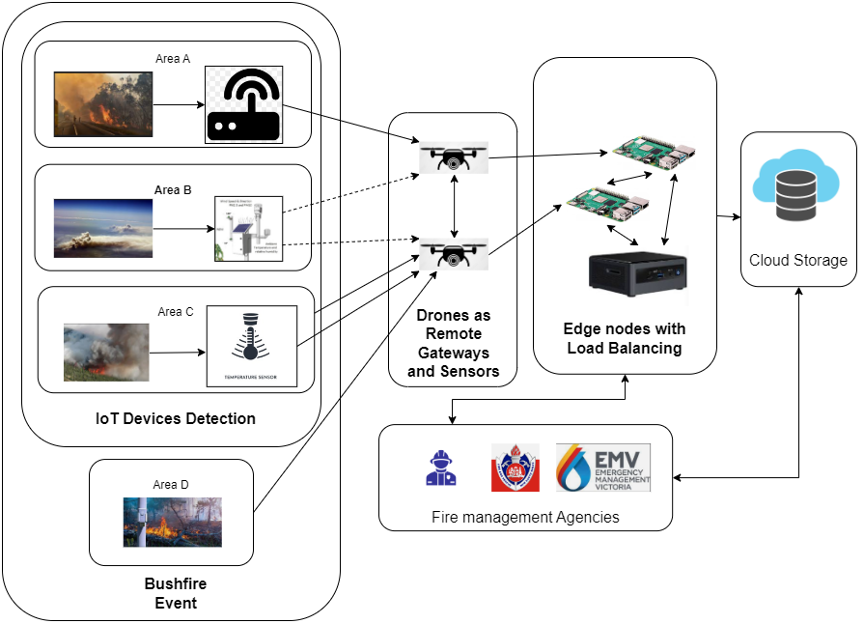}
\caption{Three-layer service architecture: sensors generate requests, UAVs relay data, and edge nodes process fire detection tasks.}\label{fig:system}
\end{figure}

We make the following assumptions. Sensor locations are fixed and known a priori. During normal operations, UAVs follow predetermined patrol routes, collecting data from sensors within range and delivering it to assigned edge nodes. When a sensor detects abnormal readings indicating a potential fire, it generates an urgent request that triggers dynamic UAV rerouting. After handling the emergency, the UAV returns to its normal route at the nearest waypoint. Table~\ref{tab:notation} summarizes the key notations.

\begin{table}[!t]
\centering
\caption{Summary of Key Notations}
\label{tab:notation}
\scriptsize
\renewcommand{\arraystretch}{1.05}
\begin{tabular}{p{1.2cm}p{6.3cm}}
\hline
\textbf{Symbol} & \textbf{Description} \\
\hline
\multicolumn{2}{l}{\textit{Sets \& Indices}} \\
$\mathbb{S}$, $n$ & Set of ground sensors, number of sensors \\
$\mathbb{G}$, $m$ & Set of UAVs, number of UAVs \\
$\mathbb{E}$, $p$ & Set of edge service nodes, number of edge nodes \\
$\mathcal{A}$ & Set of all service requests \\
$\mathbb{S}_{direct}$ & Sensors with direct edge connectivity \\
$\mathbb{S}_{UAV}$ & Sensors requiring UAV-mediated service \\
$\mathcal{A}_{urgent}$ & Set of urgent service requests \\
\hline
\multicolumn{2}{l}{\textit{Network \& Communication Parameters}} \\
$A_{monitor}$ & Monitoring area (km$^2$) \\
$r_s$, $r_g$, $r_e$ & Communication range of sensors, UAVs, edge nodes \\
$r_{sg}$, $r_{ge}$, $r_{se}$ & Link ranges: sensor-gateway, gateway-edge, sensor-edge \\
$dr$ & Link data rate (Mbps) \\
$d_{x,y}$ & Euclidean distance between entities $x$ and $y$ \\
\hline
\multicolumn{2}{l}{\textit{Service Request Parameters}} \\
$a_i$ & Service request of sensor $s_i$ \\
$\alpha_i$ & Data size of request $a_i$ (MB) \\
$\beta_i$ & Computational requirement of $a_i$ (MI) \\
\hline
\multicolumn{2}{l}{\textit{UAV Parameters}} \\
$v_g$ & UAV flight speed (m/s) \\
$P_{fly}$ & Flight power consumption (W) \\
$P_{comm}$ & Communication power consumption (W) \\
$E_{max}$ & Battery capacity (Wh) \\
$E_j$ & Energy consumption of UAV $g_j$ (Wh) \\
$L_j$ & Patrol route length of UAV $g_j$ (m) \\
$M_{max}$ & Maximum available fleet size \\
\hline
\multicolumn{2}{l}{\textit{Edge Service Parameters}} \\
$C_{e_k}$ & Service capacity of edge node $e_k$ (MIPS) \\
$Load_{e_k}$ & Current service load on $e_k$ \\
$\theta_{max}$ & Maximum service utilization threshold \\
\hline
\multicolumn{2}{l}{\textit{Time Components}} \\
$T^{total}$ & End-to-end service response time \\
$T^{lat}$, $T^{tra}$ & Communication latency, transmission time \\
$T^{exe}$, $T^{wait}$ & Execution time, waiting time \\
$T^{moving}$ & UAV travel time to edge node \\
$T_{period}$ & Service monitoring period (s) \\
$T_{max}$ & Maximum revisit time constraint (s) \\
$T_{urgent}$ & Urgent service deadline (s) \\
$t_r$ & Revisit period (s) \\
$t_{comm}$ & Communication window duration (s) \\
\hline
\multicolumn{2}{l}{\textit{Algorithm Parameters}} \\
$h_{s_i}$ & Historical fire count at sensor $s_i$ \\
$w_{s_i}$ & Sensor weight based on fire history \\
$\omega_h$ & Fire history weight coefficient \\
$\omega_d$, $\omega_l$ & Distance and load weighting coefficients \\
$\lambda$ & Route-computation trade-off parameter \\
$\epsilon$ & K-means convergence threshold \\
$D_{i,j}$ & Weighted distance (sensor $s_i$ to center $c_j$) \\
$c_j$ & Cluster center of cluster $j$ \\
\hline
\end{tabular}
\end{table}

\subsection{Service Model}

Let $\mathbb{S} = \{s_1, \ldots, s_n\}$ denote $n$ ground sensors deployed over a monitoring area of $A_{monitor}$ km$^2$, $\mathbb{G} = \{g_1, \ldots, g_m\}$ denote $m$ UAVs, and $\mathbb{E} = \{e_1, \ldots, e_p\}$ denote $p$ edge nodes. The effective communication range for each link type is determined by the shorter-range endpoint: $r_{sg} = \min\{r_g, r_s\}$, $r_{ge} = \min\{r_g, r_e\}$, and $r_{se} = \min\{r_s, r_e\}$.

Each sensor $s_i$ generates one service request $a_i = (\alpha_i, \beta_i)$ per monitoring period $T_{period}$, where $\alpha_i$ is the data size (MB) and $\beta_i$ is the computational requirement (MI). Sensors are partitioned into two groups based on edge reachability. Those within direct range of at least one edge node form $\mathbb{S}_{direct} = \{s_i : \exists\, e_k,\; d_{s_i,e_k} \leq r_{se}\}$ and can transmit data without UAV mediation. The remaining sensors $\mathbb{S}_{UAV} = \mathbb{S} \setminus \mathbb{S}_{direct}$ require UAV relay. Note that although sensors in $\mathbb{S}_{direct}$ bypass UAV routing, their requests still consume edge computational resources. Our edge assignment procedure (Section~\ref{sec:edge_assignment}) explicitly accounts for this coupling.

\subsection{Service Response Time Model}

The end-to-end response time for a service request $a_i$ consists of five components:
\begin{equation}
\label{eq:total_time}
T^{total}_{a_i} = T^{lat} + T^{tra} + T^{exe} + T^{wait} + T^{moving}
\end{equation}

\textit{Communication latency} $T^{lat}$ depends on the delivery path. For direct sensor-to-edge communication, $T^{lat} = t^{lat}_{s_i,e_k}$. For UAV-mediated delivery, $T^{lat} = t^{lat}_{s_i,g_j} + t^{lat}_{g_j,e_k}$, reflecting the two-hop relay. \textit{Transmission time} for uploading data of size $\alpha_i$ at rate $dr$ is:
\begin{equation}
\label{eq:transmission_time}
T^{tra} = \frac{\alpha_i}{dr}
\end{equation}

\textit{Execution time} at edge node $e_k$ with processing capacity $C_{e_k}$ (MIPS) is:
\begin{equation}
\label{eq:execution_time}
T^{exe} = \frac{\beta_i}{C_{e_k}}
\end{equation}

\textit{Waiting time} captures the delay before a UAV comes within range. Assuming a UAV traverses its patrol route of length $L_j$ at constant speed $v_g$, the revisit period is $t_r = L_j / v_g$ and the communication window at each sensor is $t_{comm} = 2r_{sg}/v_g$. Under a uniform arrival model, the expected waiting time is:
\begin{equation}
\label{eq:waiting_time}
\mathbb{E}[T^{wait}] = \frac{t_r - t_{comm}}{2}
\end{equation}

\textit{Travel time} for the UAV to reach the assigned edge node after data collection is:
\begin{equation}
\label{eq:moving_time}
T^{moving} = \frac{d_{g_j,e_k}}{v_g}
\end{equation}

Among these, $T^{tra}$ and $T^{exe}$ are determined by request characteristics and edge capacity, while $T^{wait}$ and $T^{moving}$ are directly influenced by routing and clustering decisions. This distinction motivates our objective formulation.

\subsection{Optimization Formulation}

We minimize a weighted combination of total patrol route length (which serves as a proxy for $T^{wait}$ and $T^{moving}$) and per-request processing cost:
\begin{equation}
\label{eq:objective}
\min \sum_{j=1}^{m} L_j + \lambda \sum_{a_i \in \mathcal{A}} (T^{tra}_i + T^{exe}_i)
\end{equation}
where $L_j$ is the patrol route length of UAV $g_j$ and $\lambda$ balances routing efficiency against edge processing cost.

This minimization is subject to five constraints. The \textit{revisit time constraint} ensures each sensor is visited frequently enough for timely data collection:
\begin{equation}
\label{eq:revisit_constraint}
t_{r,j} = \frac{L_j}{v_g} \leq T_{max}, \quad \forall g_j \in \mathbb{G}
\end{equation}

The \textit{energy constraint} ensures each UAV can complete its patrol on a single battery charge, accounting for both flight energy and communication energy:
\begin{equation}
\label{eq:energy_constraint}
\scalebox{0.8}{$\displaystyle  
E_j = P_{fly} \cdot \frac{L_j}{v_g} + P_{comm} \cdot \sum_{s_i \in Cluster_j \cap \mathbb{S}_{UAV}} \frac{\alpha_i}{dr} \leq E_{max}, \quad \forall g_j \in \mathbb{G}
$}
\end{equation}

The \textit{capacity constraint} prevents any edge node from being overloaded:
\begin{equation}
\label{eq:computing_constraint}
\sum_{a_i: assigned\_to\_e_k} \frac{\beta_i}{T_{period}} \leq C_{e_k}, \quad \forall e_k \in \mathbb{E}
\end{equation}

The \textit{deadline constraint} guarantees timely emergency response:
\begin{equation}
\label{eq:urgent_constraint}
T^{total}_{a_i} \leq T_{urgent}, \quad \forall a_i \in \mathcal{A}_{urgent}
\end{equation}

The \textit{fleet size constraint} caps the number of deployed UAVs:
\begin{equation}
\label{eq:fleet_constraint}
m \leq M_{max}
\end{equation}

The decision variables are the number of UAVs $m$, the assignment of sensors to clusters, the mapping of clusters to edge nodes, and the patrol route for each UAV. This problem is NP-hard, as it embeds multi-depot TSP with additional resource and timing constraints. We present our solution approach next.

\section{PROPOSED SERVICE FRAMEWORK}

The framework operates in two phases. In the \textit{offline planning} phase, we jointly determine sensor clustering, edge node assignment, patrol routes, and fleet size (Sections~\ref{sec:clustering}--\ref{sec:algorithm}). All constraints (Eqs.~\ref{eq:revisit_constraint}--\ref{eq:fleet_constraint}) are validated within an iterative procedure that incrementally adjusts fleet size until feasibility is achieved. In the \textit{online execution} phase, UAVs follow planned routes and respond to fire emergencies via dynamic rerouting (Section~\ref{sec:emergency}). Fig.~\ref{fig:system} illustrates the overall architecture.

\subsection{Fire-History-Weighted Clustering}
\label{sec:clustering}

The first step partitions sensors in $\mathbb{S}_{UAV}$ into $m$ clusters to enable balanced workload distribution and compact patrol routes. Standard K-means treats all sensors equally; we modify it to prioritize high-risk areas by assigning each sensor a weight based on its fire history:
\begin{equation}
\label{eq:sensor_weight}
w_{s_i} = 1 + \omega_h \cdot h_{s_i}
\end{equation}
where $h_{s_i}$ is the historical fire count at sensor $s_i$ and $\omega_h$ controls the influence of fire history. The baseline weight of 1 ensures that sensors without fire history are still included.

\textbf{Initialization.} To encourage cluster centers to form near edge nodes and thereby reduce $T^{moving}$ for data delivery, we initialize centers at edge node positions. When $m \leq p$, centers are placed at $m$ randomly selected edge nodes. When $m > p$, all $p$ edge positions are used and the remaining $m - p$ centers are initialized at the positions of the highest-weight sensors.

\textbf{Assignment.} Sensors are assigned to the nearest center under a weighted distance metric:
\begin{equation}
\label{eq:weighted_distance}
D_{i,j} = d_{s_i,c_j} \cdot \left(2 - \frac{w_{s_i}}{w_{max}}\right)
\end{equation}
where $w_{max} = \max_i w_{s_i}$. The scaling factor $(2 - w_{s_i}/w_{max})$ ranges from 1 (highest-risk) to 2 (no history), so high-risk sensors have effectively shorter distances to all centers. This draws cluster centers toward high-risk regions, resulting in denser coverage and more frequent visits.

\textbf{Update.} Centroids are recomputed as weighted averages:
\begin{equation}
\label{eq:centroid_update}
c_j = \frac{\sum_{s_i \in Cluster_j} w_{s_i} \cdot pos_{s_i}}{\sum_{s_i \in Cluster_j} w_{s_i}}
\end{equation}
The algorithm iterates until all center movements fall below a threshold $\epsilon$.

\subsection{Edge Service Assignment}
\label{sec:edge_assignment}

After clustering, each sensor's data must be routed to an edge node for processing. We perform this assignment in two phases to account for the distinct load patterns of direct and UAV-mediated sensors.

\textbf{Phase 1: Direct-communication sensors.} Each sensor $s_i \in \mathbb{S}_{direct}$ is assigned to its nearest reachable edge node:
\begin{equation}
\label{eq:direct_assignment}
e_i^* = \arg\min_{e_k:\, d_{s_i,e_k} \leq r_{se}} d_{s_i,e_k}
\end{equation}
This phase must precede cluster assignment because direct sensors impose a fixed, non-negotiable load on nearby edge nodes. Without accounting for this load first, cluster assignment may overcommit edge resources.

\textbf{Phase 2: UAV-mediated clusters.} Each cluster $j$ is assigned to an edge node by jointly minimizing distance and load:
\begin{equation}
\label{eq:cluster_assignment}
e_j^* = \arg\min_{e_k \in \mathbb{E}} \left(\omega_d \cdot \bar{d}_{j,k} + \omega_l \cdot \frac{Load_{e_k} + \beta_j}{C_{e_k}}\right)
\end{equation}
where $\bar{d}_{j,k}$ is the average distance from sensors in cluster $j$ to edge node $e_k$, $\beta_j = \sum_{s_i \in Cluster_j} \beta_i / T_{period}$ is the cluster's aggregate computational demand, and $\omega_d$, $\omega_l$ balance proximity against load. After each assignment, the load is updated: $Load_{e_j^*} \leftarrow Load_{e_j^*} + \beta_j$. If any edge node becomes overloaded ($Load_{e_k} > C_{e_k}$), we apply greedy reassignment: the least-demanding cluster on the overloaded node is moved to the next-best available node.

\subsection{Route Optimization}
\label{sec:routing}

For each cluster $j$, we solve a TSP over the sensors in $Cluster_j \cap \mathbb{S}_{UAV}$ to determine the patrol route, starting and ending at the assigned edge node $e_j^*$. Minimizing route length directly reduces both the revisit period $t_r$ and the expected waiting time $\mathbb{E}[T^{wait}]$.

We construct an initial tour using the nearest-neighbor heuristic, then apply 2-opt local search. At each iteration, for every pair of non-adjacent edges $(i, i{+}1)$ and $(k, k{+}1)$ in the tour, we check whether reversing the segment $[i{+}1, k]$ yields a shorter tour:
\begin{equation}
\label{eq:2opt_condition}
d_{i,k} + d_{i+1,k+1} < d_{i,i+1} + d_{k,k+1}
\end{equation}
If so, the reversal is applied. The procedure repeats until no improving swap exists. The total energy consumption of UAV $g_j$ is then:
\begin{equation}
\label{eq:energy_calculation}
E_j = P_{fly} \cdot \frac{L_j}{v_g} + P_{comm} \cdot \sum_{s_i \in Cluster_j \cap \mathbb{S}_{UAV}} \frac{\alpha_i}{dr}
\end{equation}

\subsection{Integrated Planning Algorithm}
\label{sec:algorithm}

The three preceding components (clustering, edge assignment, and routing) are interdependent. Clustering determines the sensor groups that define both patrol routes and edge workloads; edge capacity constraints may in turn require re-clustering with more UAVs. Algorithm~\ref{alg:integrated} integrates these steps within an iterative fleet-sizing loop that begins with a lower-bound estimate $m = \lceil A_{monitor} / (\pi r_{sg}^2) \rceil$ and increments $m$ until all constraints are satisfied or $M_{max}$ is exceeded.

\begin{algorithm}[!t]
\scriptsize
\caption{Integrated Service Planning Algorithm}
\label{alg:integrated}
\begin{algorithmic}[1]
\Require Sensors $\mathbb{S}$, edges $\mathbb{E}$, area $A_{monitor}$, fire history $\{h_i\}$, parameters
\Ensure UAV count $m^*$, routes $\{Path_j\}$, edge service assignments
\State $m \leftarrow \lceil A_{monitor} / (\pi r_{sg}^2) \rceil$
\State $feasible \leftarrow \textsc{False}$
\While{$\neg feasible$ \textbf{and} $m \leq M_{max}$}
    \State \textit{// Identify direct communication sensors}
    \State $\mathbb{S}_{direct} \leftarrow \{s_i : \exists e_k, d_{s_i,e_k} \leq r_{se}\}$
    \State $\mathbb{S}_{UAV} \leftarrow \mathbb{S} \setminus \mathbb{S}_{direct}$
    \State \textit{// Weighted K-means clustering on $\mathbb{S}_{UAV}$}
    \State Compute sensor weights: $w_{s_i} \leftarrow 1 + \omega_h \cdot h_{s_i}$, $\forall s_i \in \mathbb{S}_{UAV}$
    \If{$m \leq p$}
        \State Initialize centers at positions of $m$ randomly selected edge nodes
    \Else
        \State Initialize centers at all $p$ edge positions and $m-p$ high-weight sensor positions
    \EndIf
    \Repeat
        \State Assign sensors to nearest weighted centers using $D_{i,j} = d_{s_i,c_j} \cdot (2 - \frac{w_{s_i}}{w_{max}})$
        \State Update centers as weighted centroids: $c_j \leftarrow \frac{\sum_{s_i \in Cluster_j} w_{s_i} \cdot pos_{s_i}}{\sum_{s_i \in Cluster_j} w_{s_i}}$
    \Until{center movement $< \epsilon$}
    \State \textit{// Edge service assignment - Phase 1: Direct sensors}
    \State Initialize $Load_{e_k} \leftarrow 0$, $\forall e_k \in \mathbb{E}$
    \For{$s_i \in \mathbb{S}_{direct}$}
        \State $e_i^* \leftarrow \arg\min_{e_k: d_{s_i,e_k} \leq r_{se}} d_{s_i,e_k}$
        \State $Load_{e_i^*} \leftarrow Load_{e_i^*} + \beta_i / T_{period}$
    \EndFor
    \State \textit{// Edge service assignment - Phase 2: UAV clusters}
    \For{$j = 1$ \textbf{to} $m$}
        \State Compute $\beta_j \leftarrow \sum_{s_i \in Cluster_j} \beta_i / T_{period}$
        \State $e_j^* \leftarrow \arg\min_{e_k} (\omega_d \bar{d}_{j,k} + \omega_l \frac{Load_{e_k} + \beta_j}{C_{e_k}})$
        \State Update $Load_{e_j^*} \leftarrow Load_{e_j^*} + \beta_j$
    \EndFor
    \State \textit{// Validate service capacity constraint}
    \If{$\exists e_k: Load_{e_k} > C_{e_k}$}
        \State Reassign overloaded clusters greedily
        \If{reassignment fails}
            \State $m \leftarrow m + 1$; \textbf{continue}
        \EndIf
    \EndIf
    \State $feasible \leftarrow \textsc{True}$
    \For{$j = 1$ \textbf{to} $m$}
        \State $Cluster_j^{route} \leftarrow Cluster_j \cap \mathbb{S}_{UAV}$ \textit{// Exclude direct sensors}
        \State \textit{// Greedy TSP construction}
        \State $Path_j \leftarrow \textsc{NearestNeighbor}(Cluster_j^{route}, e_j^*)$
        \State \textit{// 2-opt improvement}
        \Repeat
            \State Try all edge swaps: if $d_{i,k} + d_{i+1,k+1} < d_{i,i+1} + d_{k,k+1}$, reverse segment $[i+1, k]$
        \Until{no improvement}
        \State Compute $L_j$, $t_{r,j} \leftarrow L_j / v_g$
        \State Compute $E_j \leftarrow P_{fly} \cdot t_{r,j} + P_{comm} \cdot \sum_{s_i \in Cluster_j^{route}} \frac{\alpha_i}{dr}$
        \State \textit{// Validate time and energy constraints}
        \If{$t_{r,j} > T_{max}$ \textbf{or} $E_j > E_{max}$}
            \State $feasible \leftarrow \textsc{False}$
            \State \textbf{break}
        \EndIf
    \EndFor
    \If{$\neg feasible$}
        \State $m \leftarrow m + 1$
    \EndIf
\EndWhile
\If{$m > M_{max}$}
    \State \Return \textsc{Infeasible} \textit{// Fleet size constraint violated}
\EndIf
\State \Return $m$, $\{Path_1, \ldots, Path_m\}$, $\{e_1^*, \ldots, e_m^*\}$
\end{algorithmic}
\end{algorithm}

The time complexity is $O(M_{max} \cdot (I_{km} \cdot nm + np + I_{2opt} \cdot n^2/m))$. For typical deployments where $n \gg m, p$, the routing phase dominates at $O(M_{max} \cdot I_{2opt} \cdot n^2/m)$.

\subsection{Emergency Service Adaptation}
\label{sec:emergency}

When a sensor $s_i$ detects conditions indicating a potential fire, its request $a_i$ is elevated to urgent status ($a_i \in \mathcal{A}_{urgent}$), triggering a three-step emergency protocol.

\textbf{Step 1: UAV dispatch.} The system selects the UAV that minimizes total detour distance, measured from its current position to the alert sensor and onward to the nearest edge node:
\begin{equation}
\label{eq:emergency_uav_selection}
g^* = \arg\min_{g_j \in \mathbb{G}} \left\{d_{g_j,s_i} + d_{s_i,e_{nearest}}\right\}
\end{equation}

\textbf{Step 2: Load-aware delivery.} The dispatched UAV collects urgent data from $s_i$ and delivers it to the nearest edge node with sufficient spare capacity:
\begin{equation}
\label{eq:emergency_edge_selection}
e^* = \arg\min_{e_k:\, Load_{e_k}/C_{e_k} < \theta_{max}} d_{s_i,e_k}
\end{equation}

\textbf{Step 3: Route resumption.} After delivery, the UAV returns to its normal patrol at the nearest waypoint: $waypoint^* = \arg\min_{w \in Path_{g^*}} d_{current,w}$. This minimizes the disruption to normal monitoring coverage.

For concurrent emergencies, alerts are prioritized by fire history $h_{s_i}$ (higher history implies higher risk) and sequentially assigned to available UAVs. This protocol targets the deadline constraint (Eq.~\ref{eq:urgent_constraint}) by minimizing each component of response time: $T^{wait}$ through immediate dispatch, $T^{moving}$ through nearest-UAV selection, and $T^{exe}$ through load-aware edge selection.

During normal operations, each UAV delivers data to its pre-assigned edge node $e_j^*$. If that node is temporarily unreachable or overloaded, a fallback selection is applied:
\begin{equation}
\label{eq:temp_edge_selection}
e^{temp} = \arg\min_{e_k:\, reachable} \left(\omega_d \cdot d_{g_j,e_k} + \omega_l \cdot \frac{Load_{e_k}}{C_{e_k}}\right)
\end{equation}

\subsection{Theoretical Guarantees}
\label{sec:theory}

We establish performance bounds that exploit the specific structure of our framework, namely the weighted distance metric and the spatial compactness of clusters.

\begin{definition}[Cluster Radius]
For cluster $j$ with center $c_j$, the cluster radius is $R_j = \max_{s_i \in Cluster_j \cap \mathbb{S}_{UAV}} d_{s_i, c_j}$.
\end{definition}

\begin{theorem}[Risk-Aware Coverage Improvement]
\label{thm:clustering}
Let $\mathbb{S}_H = \{s_i \in \mathbb{S}_{UAV} : h_{s_i} > 0\}$ denote sensors with fire history. Under the weighted distance metric (Eq.~\ref{eq:weighted_distance}), let $\bar{R}_H^{w}$ and $\bar{R}_H^{u}$ denote the average cluster radius experienced by sensors in $\mathbb{S}_H$ under weighted and standard K-means, respectively. Then:
\begin{equation}
\label{eq:radius_reduction}
\bar{R}_H^{w} \leq \frac{2}{1 + \bar{w}_H / w_{max}} \cdot \bar{R}_H^{u}
\end{equation}
where $\bar{w}_H = |\mathbb{S}_H|^{-1}\sum_{s_i \in \mathbb{S}_H} w_{s_i}$ is the average weight of high-risk sensors and $w_{max} = \max_i w_{s_i}$.
\end{theorem}

\begin{proof}
From Eqs.~\eqref{eq:sensor_weight} and~\eqref{eq:weighted_distance}, sensor $s_i$ has effective distance $D_{i,j} = d_{s_i,c_j} \cdot (2 - w_{s_i}/w_{max})$. For $s_i \in \mathbb{S}_H$, $w_{s_i} > 1$, so the scaling factor $(2 - w_{s_i}/w_{max}) < (2 - 1/w_{max})$, which is the factor for zero-history sensors. High-risk sensors thus experience strictly shorter effective distances, causing preferential assignment to nearby centers. The weighted centroid update (Eq.~\ref{eq:centroid_update}) further shifts centers toward high-risk sensors. Averaging the scaling factor $(2 - w_{s_i}/w_{max})$ over $\mathbb{S}_H$ and applying the identity $2 - \bar{w}_H/w_{max} \leq 2/(1 + \bar{w}_H/w_{max})$ for $\bar{w}_H \leq w_{max}$ yields \eqref{eq:radius_reduction}.
\end{proof}

\textit{Remark.} The reduction factor $2/(1 + \bar{w}_H/w_{max})$ is strictly less than 2 whenever fire history exists, and decreases as either $\omega_h$ or the average fire count increases. Since expected waiting time (Eq.~\ref{eq:waiting_time}) scales with route length, which in turn scales with cluster radius, this directly translates to more frequent revisits for high-risk areas. The bound also clarifies the role of $\omega_h$: increasing it strengthens prioritization, but the marginal gain diminishes as $\bar{w}_H/w_{max}$ saturates toward~1.

\begin{theorem}[Emergency Response Bound]
\label{thm:emergency}
If an emergency occurs at sensor $s_i \in \mathbb{S}_{UAV}$ belonging to cluster $j_0$, the nearest-UAV dispatch policy (Eq.~\ref{eq:emergency_uav_selection}) guarantees:
\begin{equation}
\label{eq:emergency_bound}
T_{emg} \leq \frac{2\, R_{max}}{v_g} + \frac{d_{max}^{se}}{v_g} + T^{tra} + T^{exe}
\end{equation}
where $R_{max} = \max_j R_j$ and $d_{max}^{se}$ is the maximum feasible sensor-to-edge distance.
\end{theorem}

\begin{proof}
Since $s_i$ lies within distance $R_{j_0}$ of center $c_{j_0}$, and UAV $g_{j_0}$ is always on a tour whose waypoints all lie within $R_{j_0}$ of $c_{j_0}$, the triangle inequality gives $d_{g_{j_0}, s_i} \leq 2R_{j_0}$. The dispatched UAV $g^*$ minimizes distance to $s_i$, so $d_{g^*, s_i} \leq 2R_{j_0} \leq 2R_{max}$. Adding delivery, transmission, and execution times yields the bound.
\end{proof}

\textit{Remark.} This bound replaces the dispatch term $T_{max}/2$ from a naive route-length analysis (which follows from $L_j \leq v_g T_{max}$ and the worst-case on-curve distance $L_j/2$) with $2R_{max}/v_g$. For well-clustered deployments, $R_{max} \ll v_g T_{max}/2$, yielding a significantly tighter bound. Crucially, the two theorems are linked: weighted clustering reduces $R_j$ for high-risk clusters (Theorem~\ref{thm:clustering}), precisely where emergencies are most likely, tightening the effective emergency bound beyond what $R_{max}$ alone suggests.

\section{PERFORMANCE EVALUATION}

\subsection{Experimental Setup}

We evaluate the proposed framework through simulations of forest fire monitoring scenarios. The default monitoring area is 100 km$^2$ with 200 ground sensors deployed following a non-uniform spatial distribution that reflects real forest fire patterns, where fires tend to cluster in specific high-risk areas~\cite{parisien2009environmental}. The framework is implemented in Python 3.10 using NumPy and SciPy.

System parameters follow practical specifications and prior UAV-MEC studies: sensor communication range $r_s = 500$ m, UAV range $r_g = 1000$ m, edge node range $r_e = 2000$ m, and link data rate $dr = 10$ Mbps~\cite{villarim2019evaluation, DJI_M300}. UAV specifications include flight speed $v_g = 15$ m/s, flight power $P_{fly} = 100$ W, communication power $P_{comm} = 5$ W, and battery capacity $E_{max} = 500$ Wh with maximum revisit time $T_{max} = 3600$ s~\cite{DJI_M300, Zeng2019Energy}. Each sensor generates requests with data size $\alpha_i \sim \mathcal{U}[1, 5]$ MB and computational requirement $\beta_i \sim \mathcal{U}[100, 500]$ MI~\cite{Bushnaq2021IoTJ, Mao2017Survey}. Five edge nodes are deployed with capacity $C_{e_k}$ ranging from 5000 to 10000 MIPS~\cite{NVIDIA_Jetson}. The monitoring period is $T_{period} = 3600$ s and the urgent deadline is $T_{urgent} = 300$ s.

Algorithm hyperparameters are selected via grid search. The fire history weight $\omega_h = 1.5$ is chosen from $\{0.5, 1.0, 1.5, 2.0, 2.5\}$ as the value yielding the lowest average response time. Edge assignment weights $\omega_d = 0.7$, $\omega_l = 0.3$ are selected from $\omega_d \in \{0.5, 0.6, 0.7, 0.8, 0.9\}$ (with $\omega_l = 1 - \omega_d$). The trade-off parameter $\lambda = 0.1$ is chosen from $\{0.01, 0.05, 0.1, 0.5, 1.0\}$. Other parameters: $\epsilon = 10$ m, $\theta_{max} = 0.8$, $M_{max} = 20$.

Each configuration is run 20 times with different random seeds; we report means with 95\% confidence intervals. Statistical significance is verified using paired t-tests ($p < 0.05$).

\subsection{Baseline Methods}

We compare against three methods that solve the same joint optimization problem (sensor clustering, edge assignment, routing, and fleet sizing) under identical constraints (Eqs.~\ref{eq:revisit_constraint}--\ref{eq:fleet_constraint}):
\begin{itemize}
\item \textbf{GA (Genetic Algorithm)}: Encodes sensor-to-UAV assignment and visit ordering jointly as a chromosome. Each individual represents a complete solution including clustering, routing, and edge assignment. Population size is 50 with 100 generations, using tournament selection, single-point crossover (rate 0.8), and mutation (rate 0.1). Fleet size is determined by the best feasible individual.

\item \textbf{PSO (Particle Swarm Optimization)}: Each particle encodes a continuous-relaxed sensor-to-UAV assignment, decoded into discrete clusters via nearest-integer mapping. Routes within each cluster are optimized per-particle, and edge assignment follows the same two-phase procedure as the proposed method. Swarm size is 30 with 100 iterations, inertia weight $w{=}0.7$, cognitive coefficient $c_1{=}1.5$, social coefficient $c_2{=}1.5$.

\item \textbf{Greedy}: Assigns each sensor to the nearest available UAV to form clusters, constructs routes via nearest-neighbor heuristic, and applies the same adaptive fleet sizing loop (incrementing UAV count until all constraints are satisfied). Edge assignment uses the same two-phase procedure. This baseline isolates the benefit of our weighted clustering and 2-opt routing by replacing only these two components with their simplest alternatives.
\end{itemize}

All methods share the same edge assignment module, constraint validation logic, and fleet sizing mechanism. Performance differences therefore reflect the effectiveness of the clustering and routing strategies alone.

\subsection{Average Service Response Time}

Fig.~\ref{fig:avg_time} compares average response times with 95\% confidence intervals. The proposed method achieves 288 s, outperforming GA (978 s), PSO (1089 s), and Greedy (1818 s) by 70.6\%, 73.6\%, and 84.2\% respectively. The narrow confidence interval ($\pm$12 s) reflects the deterministic nature of 2-opt convergence, whereas Greedy's wide interval ($\pm$210 s) stems from its sensitivity to initial sensor ordering. Two factors drive this gap: weighted clustering concentrates UAV coverage around high-risk sensors, reducing their waiting time, while 2-opt eliminates path crossings that inflate revisit periods.

\begin{figure}[!t]
\centering
\includegraphics[width=0.85\linewidth]{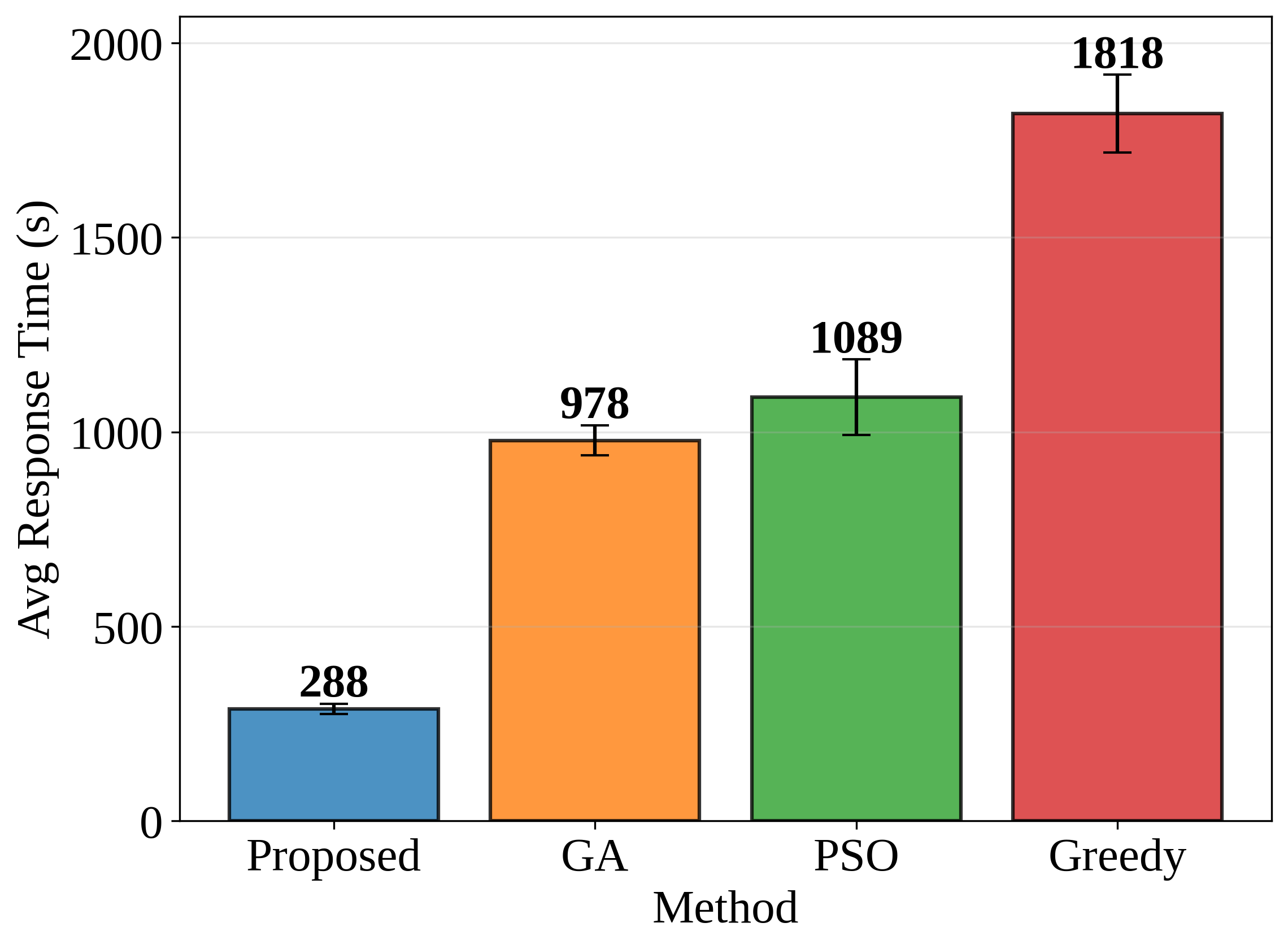}
\caption{Average service response time comparison with 95\% confidence intervals.}
\label{fig:avg_time}
\end{figure}

\subsection{Service Response Time CDF}

Fig.~\ref{fig:cdf} shows the response time distribution across all sensors. The proposed method exhibits a compact distribution with all sensors served within 150--600 s. At the 95th percentile, it achieves 500 s compared to 1100 s (GA), 1700 s (PSO), and 3800 s (Greedy). Notably, the steep CDF slope indicates that the gap between the best-served and worst-served sensors is small (under 450 s), meaning no sensor is systematically neglected. In contrast, Greedy's gradual slope spans over 4000 s, indicating that peripheral sensors in sparse regions receive disproportionately poor service. This tail behavior is particularly concerning for wildfire monitoring, where the most delayed sensors may be precisely those in remote high-risk areas.

\begin{figure}[!t]
\centering
\includegraphics[width=0.85\linewidth]{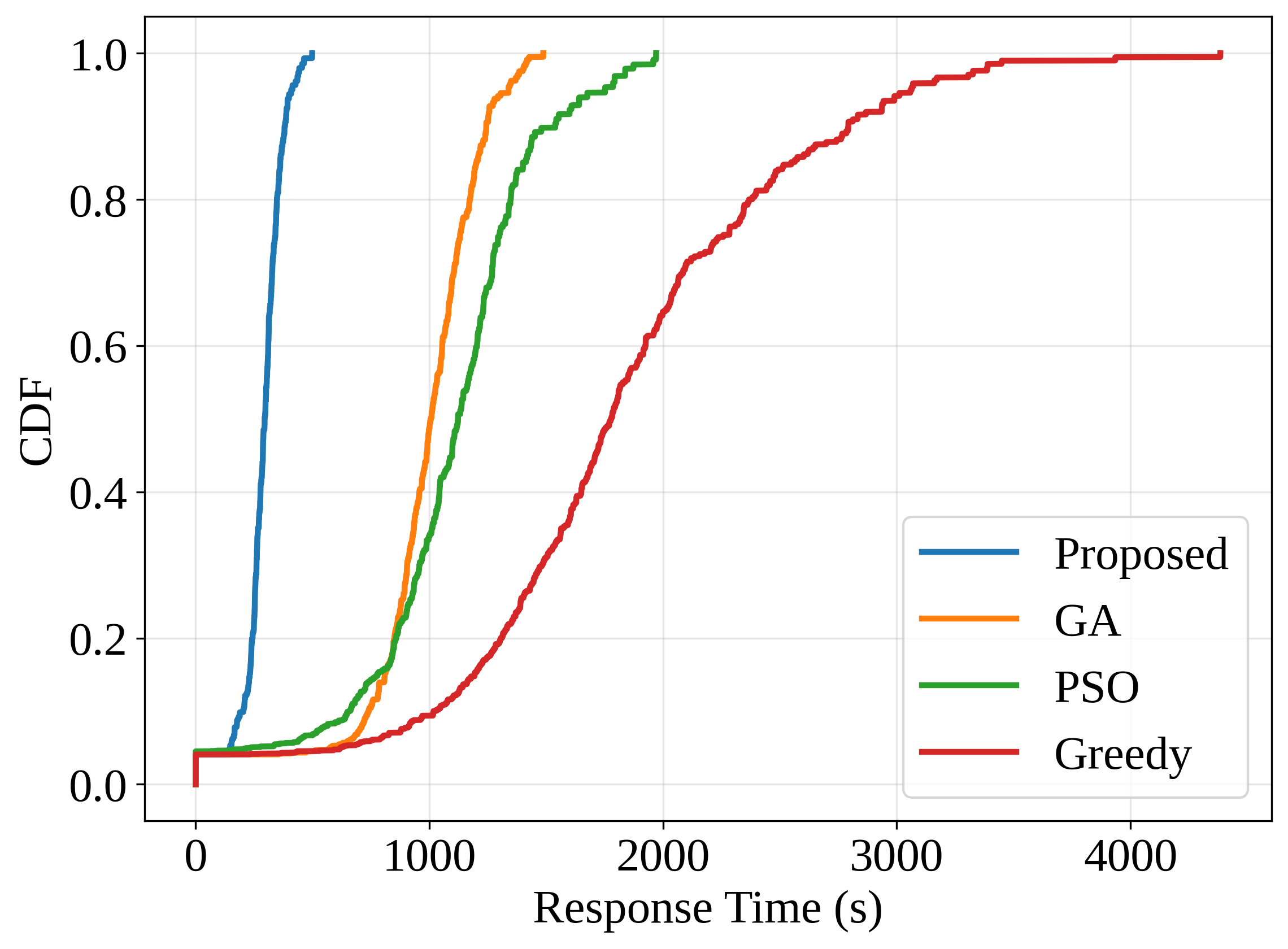}
\caption{Cumulative distribution function of service response times.}
\label{fig:cdf}
\end{figure}

\subsection{Energy Efficiency}

Fig.~\ref{fig:energy} shows total energy consumption across all UAVs per monitoring period. The proposed method consumes 217 Wh, achieving 76.0\%, 73.8\%, and 88.4\% reductions compared to GA (905 Wh), PSO (827 Wh), and Greedy (1863 Wh). Since flight power ($P_{fly} = 100$ W) dominates communication power ($P_{comm} = 5$ W) by a factor of 20, energy savings are driven almost entirely by shorter flight distances. The 114.5 km total path length of the proposed method (vs.\ 535--1047 km for baselines) translates directly into proportional energy reductions, which in turn extend operational endurance for sustained monitoring in remote areas.

\begin{figure}[!t]
\centering
\includegraphics[width=0.85\linewidth]{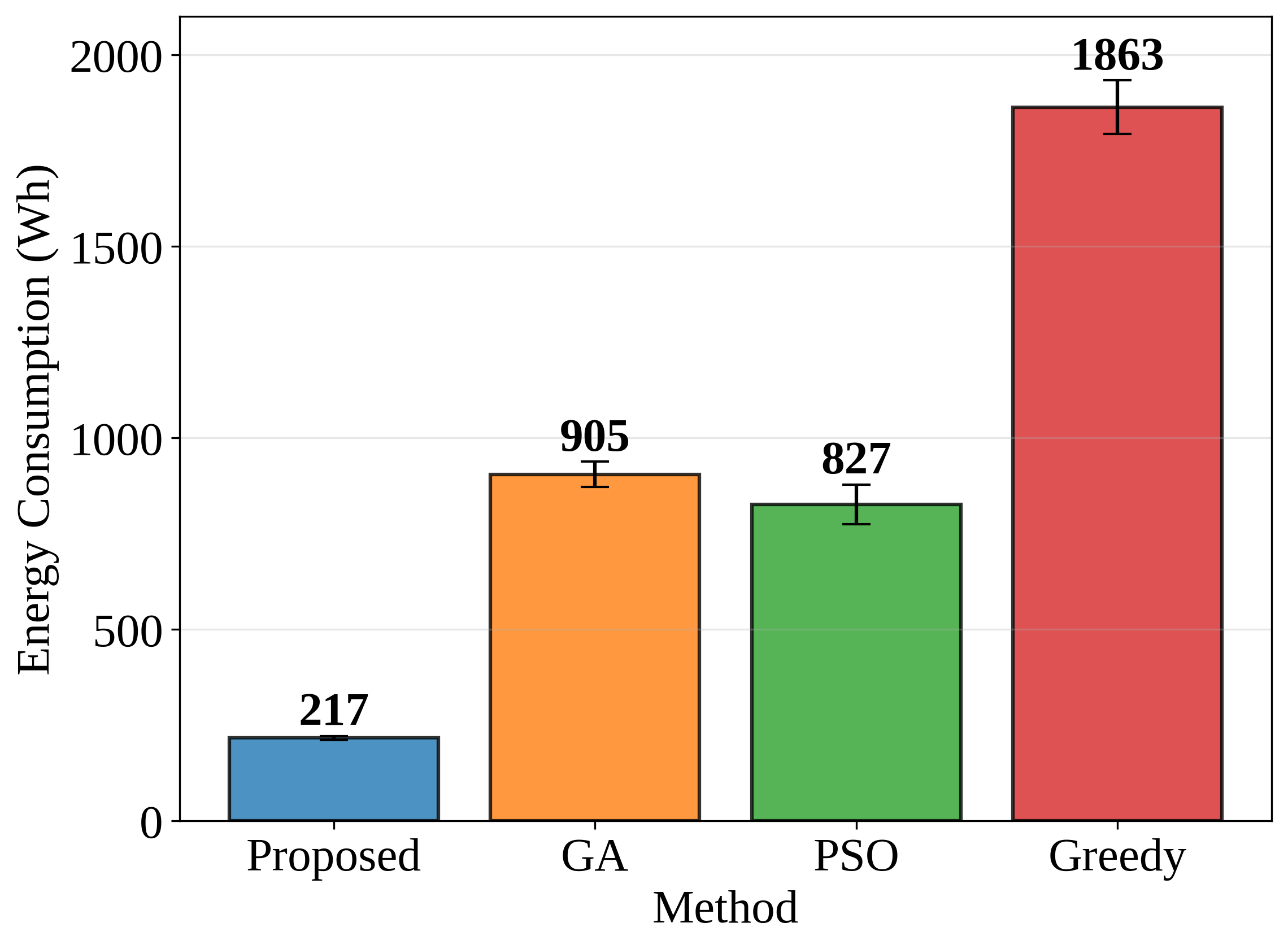}
\caption{Total energy consumption comparison.}
\label{fig:energy}
\end{figure}

\subsection{Fleet Size Requirement}

Fig.~\ref{fig:fleet} compares the number of UAVs required to satisfy all constraints (Eqs.~\ref{eq:revisit_constraint}--\ref{eq:fleet_constraint}). The proposed method requires 11 UAVs, while GA, PSO, and Greedy need 16, 15, and 19 respectively, representing reductions of 31.3\%, 26.7\%, and 42.1\%. The adaptive fleet sizing loop in Algorithm~\ref{alg:integrated} begins from a coverage-based lower bound and only adds UAVs when energy or revisit constraints are violated. Because 2-opt produces shorter routes per cluster, each UAV can cover more sensors within the energy budget, reducing the total fleet required. This translates to lower capital and operational costs for large-scale deployments.

\begin{figure}[!t]
\centering
\includegraphics[width=0.85\linewidth]{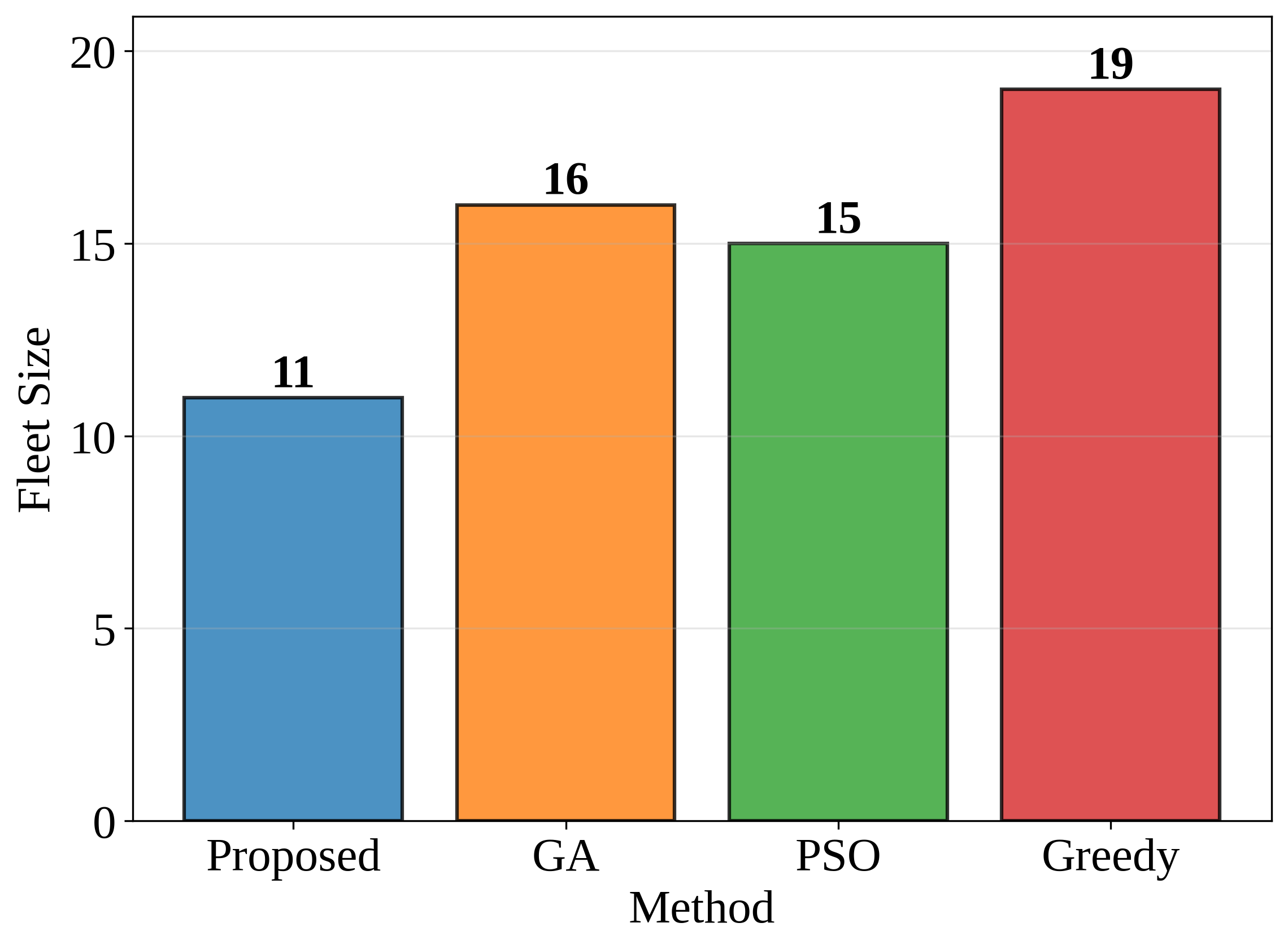}
\caption{UAV fleet size comparison.}
\label{fig:fleet}
\end{figure}

\begin{figure*}[!t]
\centering
\includegraphics[width=0.8\textwidth, height=0.6\textheight]{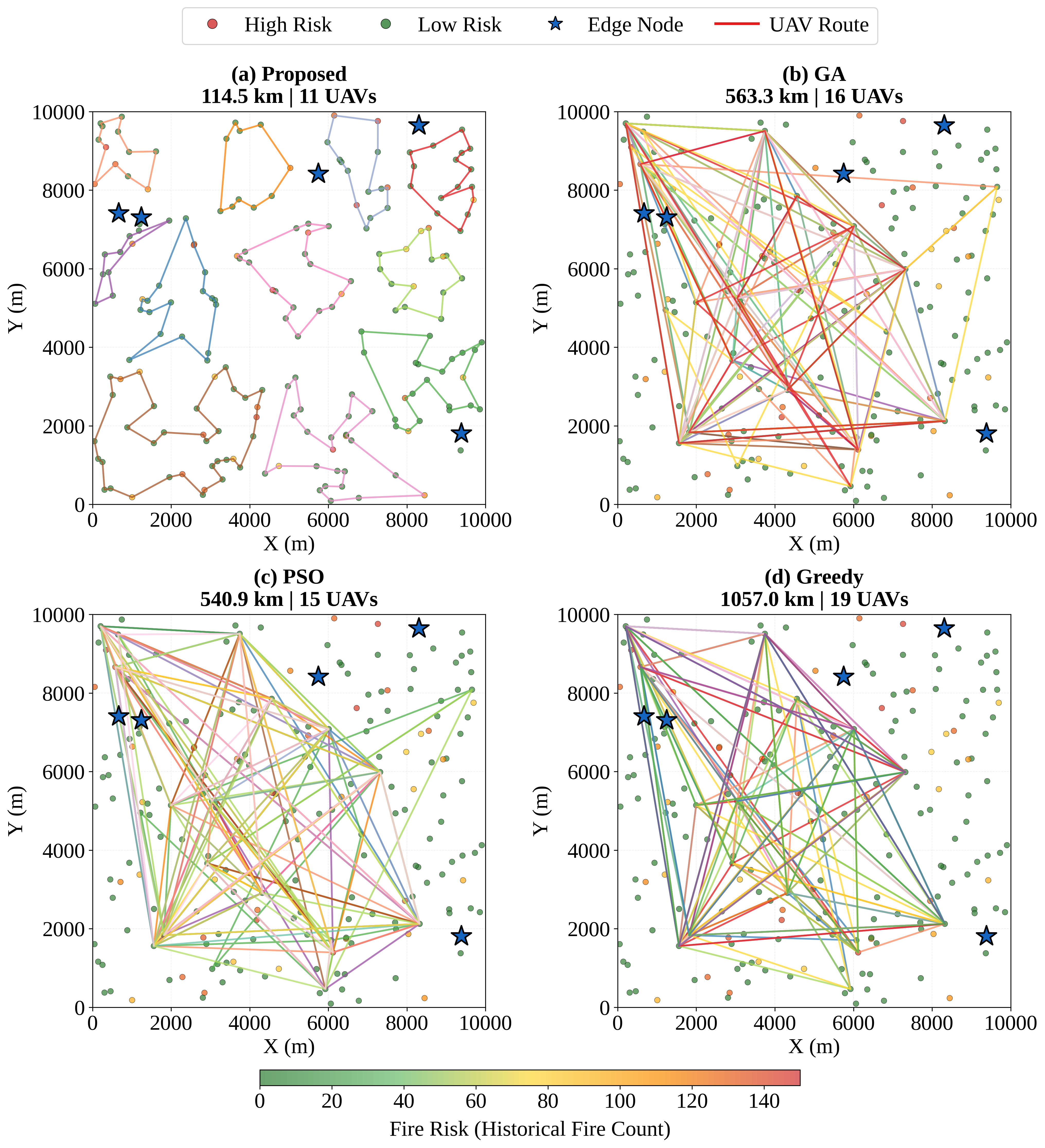}
\caption{UAV route visualization comparing spatial efficiency. (a) The proposed method achieves compact, non-overlapping routes (114.5 km, 11 UAVs). (b-d) Baseline methods exhibit significantly longer routes with extensive overlaps: GA (563.3 km, 16 UAVs), PSO (535.0 km, 15 UAVs), and Greedy (1047.0 km, 19 UAVs).}
\label{fig:route_visualization}
\end{figure*}

\subsection{Route Visualization and Spatial Analysis}

Fig.~\ref{fig:route_visualization} visualizes UAV flight paths for a representative 200-sensor scenario. Sensor colors indicate fire risk (green: low, red: high); blue stars denote edge nodes.

The proposed method produces 11 compact, well-separated routes totaling 114.5 km. Two spatial properties are worth highlighting. First, the routes are non-overlapping: weighted clustering assigns each sensor to exactly one UAV, eliminating redundant coverage. Second, high-risk sensor clusters (red) tend to have shorter routes, meaning shorter revisit periods and more frequent data collection. This is a direct consequence of the weighted distance metric (Eq.~\ref{eq:weighted_distance}), which draws cluster centers toward high-risk regions. In contrast, GA produces 563.3 km of routes (4.9$\times$) with visible path crossings, PSO requires 535.0 km (4.7$\times$), and Greedy generates 1047.0 km (9.1$\times$) with extensive inter-route overlap despite deploying more UAVs.

\subsection{Ablation Study}

To isolate the contribution of each component, we compare four configurations: Full (both K-means clustering and 2-opt), w/o 2-opt (clustering only), w/o K-means (random assignment with 2-opt), and w/o Both (random assignment, nearest-neighbor routing).

\begin{figure*}[!t]
\centering
\begin{subfigure}[b]{0.32\textwidth}
\centering
\includegraphics[width=\textwidth]{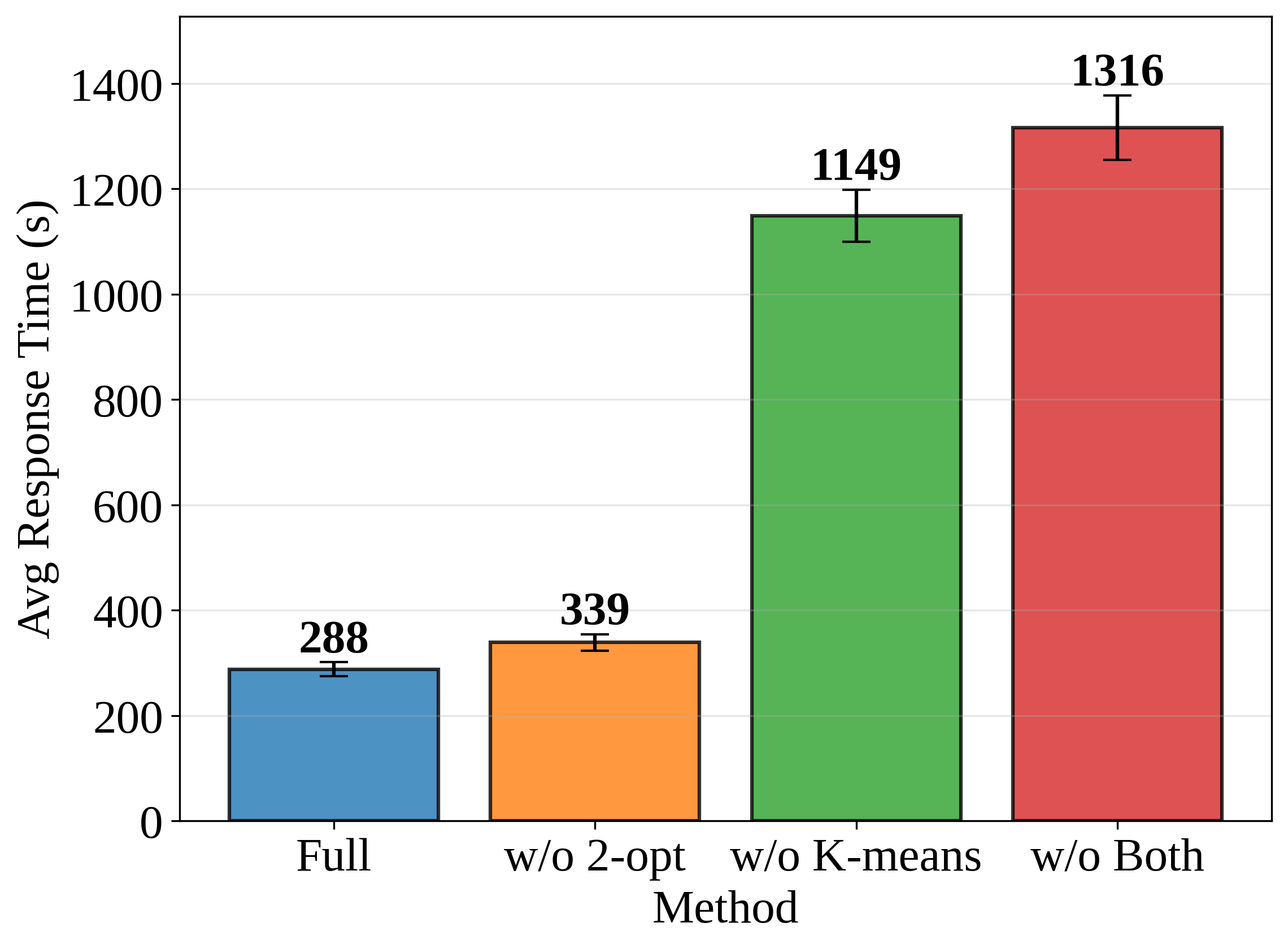}
\caption{Average response time}
\label{fig:ablation_avg}
\end{subfigure}
\hfill
\begin{subfigure}[b]{0.32\textwidth}
\centering
\includegraphics[width=\textwidth]{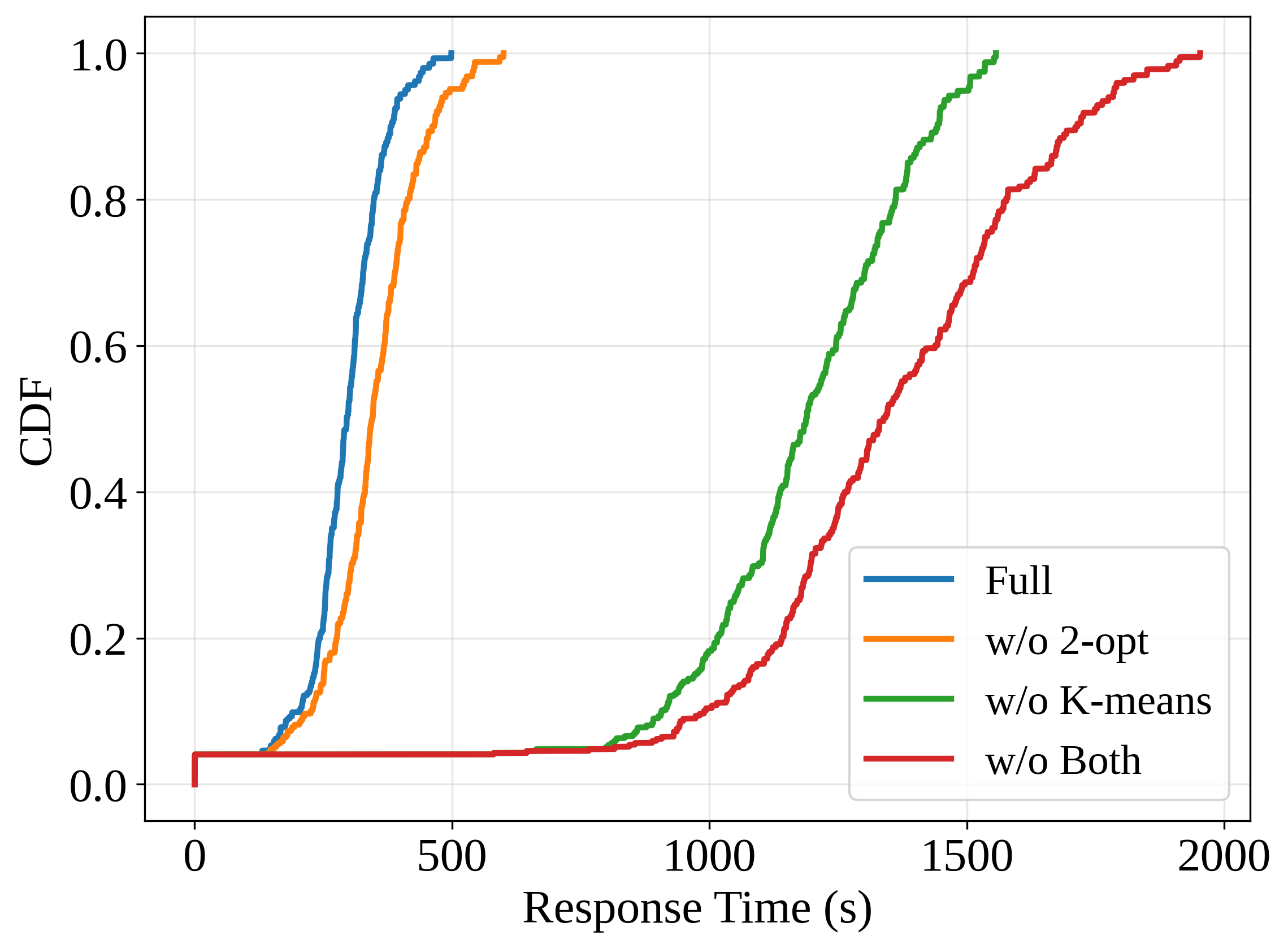}
\caption{Response time CDF}
\label{fig:ablation_cdf}
\end{subfigure}
\hfill
\begin{subfigure}[b]{0.32\textwidth}
\centering
\includegraphics[width=\textwidth]{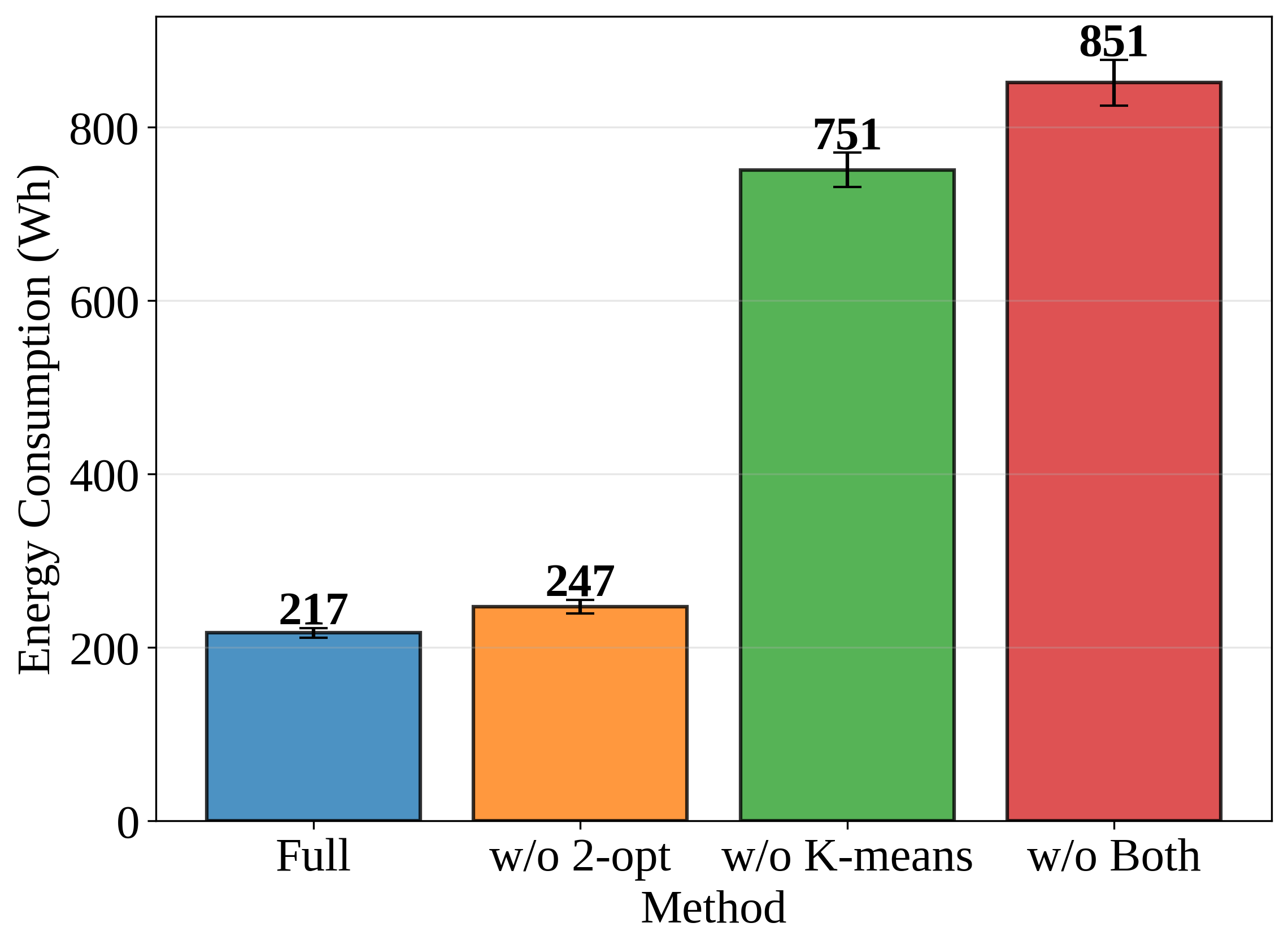}
\caption{Energy consumption}
\label{fig:ablation_energy}
\end{subfigure}
\caption{Ablation study results comparing the contribution of K-means clustering and 2-opt route optimization.}
\label{fig:ablation}
\end{figure*}

\begin{figure*}[!t]
\centering
\begin{subfigure}[b]{0.48\textwidth}
\centering
\includegraphics[width=\textwidth]{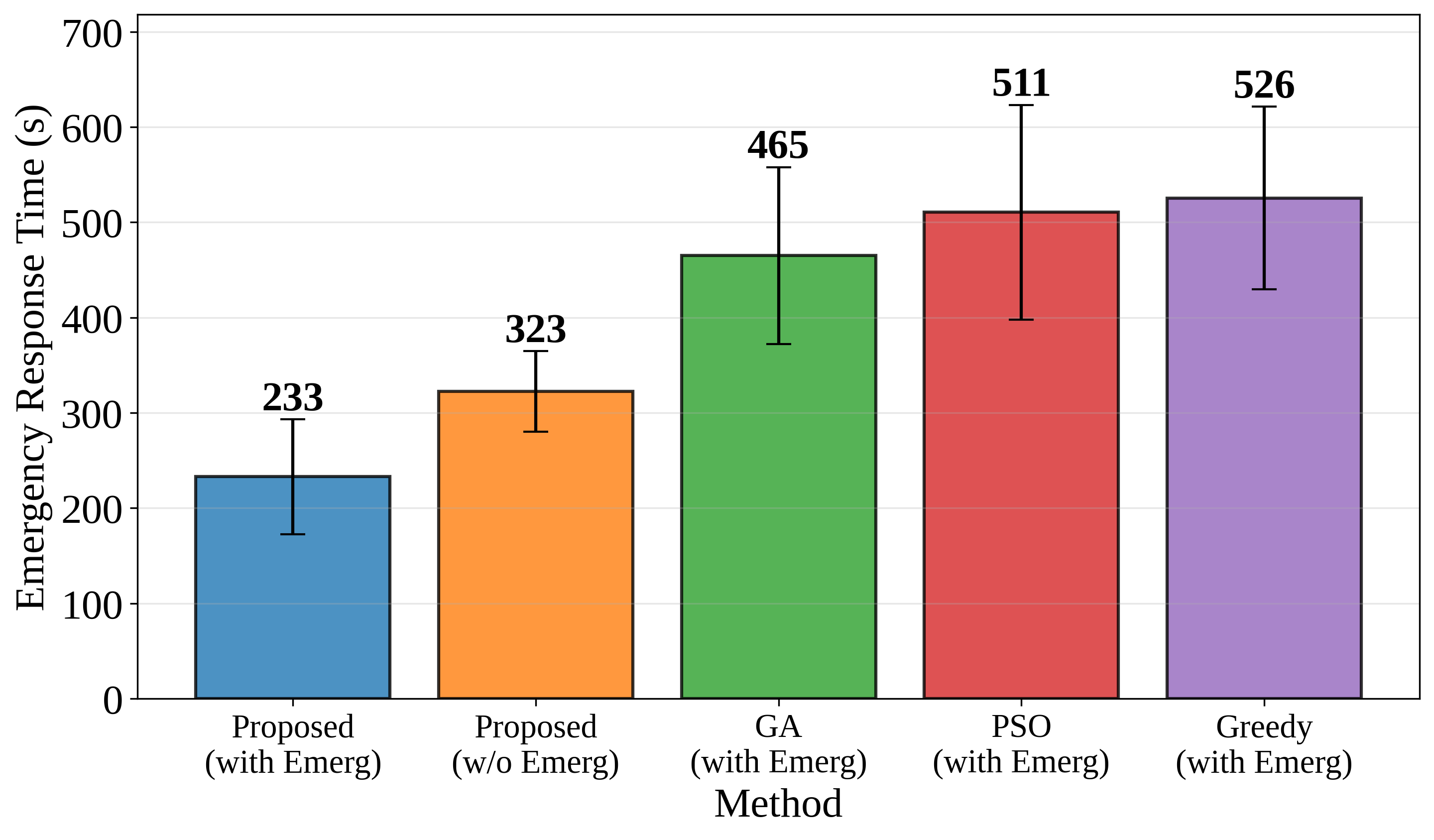}
\caption{Emergency service response time}
\label{fig:emergency_avg}
\end{subfigure}
\hfill
\begin{subfigure}[b]{0.48\textwidth}
\centering
\includegraphics[width=\textwidth]{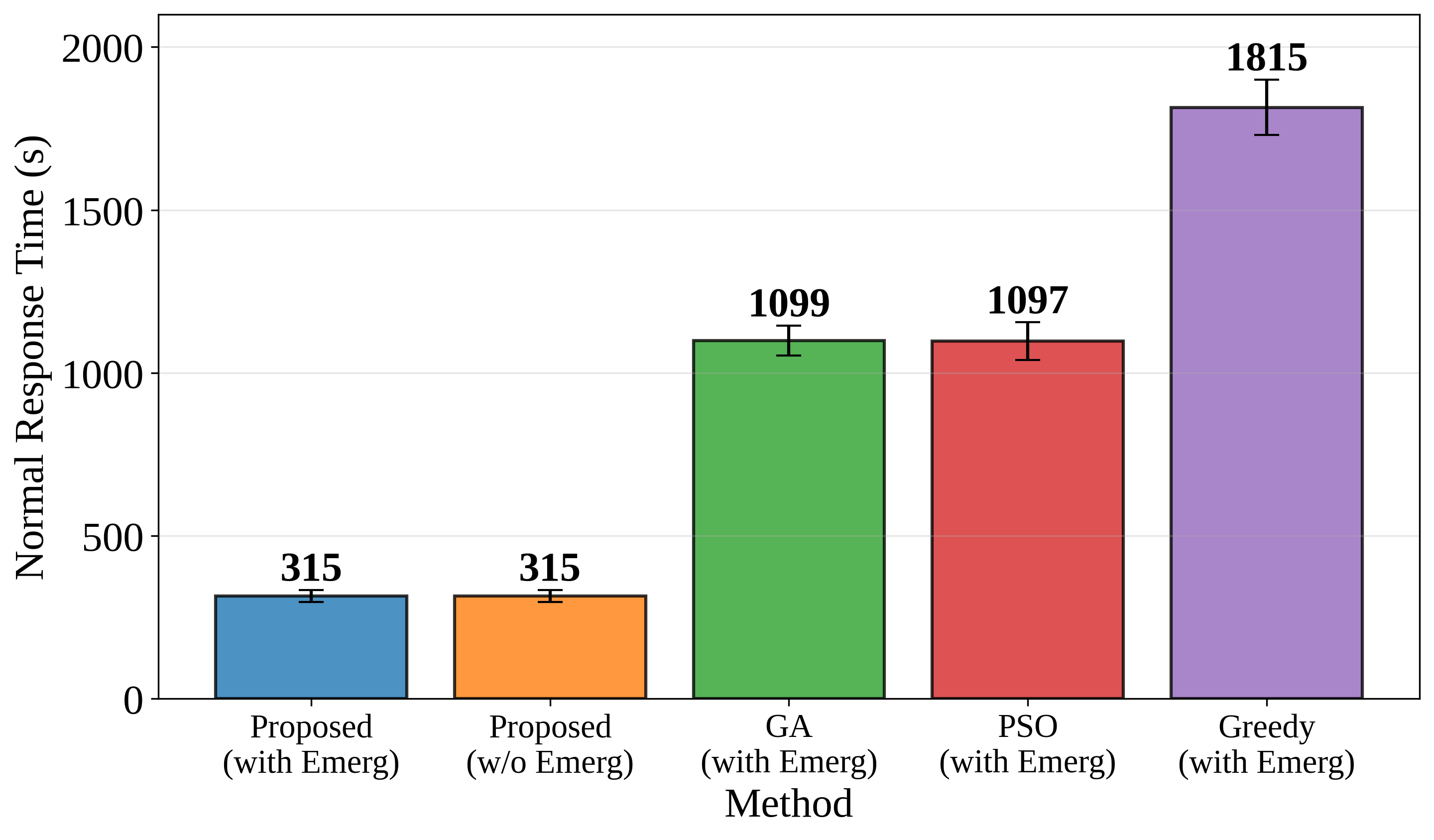}
\caption{Normal service response time}
\label{fig:emergency_normal}
\end{subfigure}
\caption{Emergency service performance comparison showing (a) response time for emergency events and (b) impact on normal service operations.}
\label{fig:emergency}
\end{figure*}

\begin{figure*}[!t]
\centering
\begin{subfigure}[b]{0.32\textwidth}
\centering
\includegraphics[width=\textwidth]{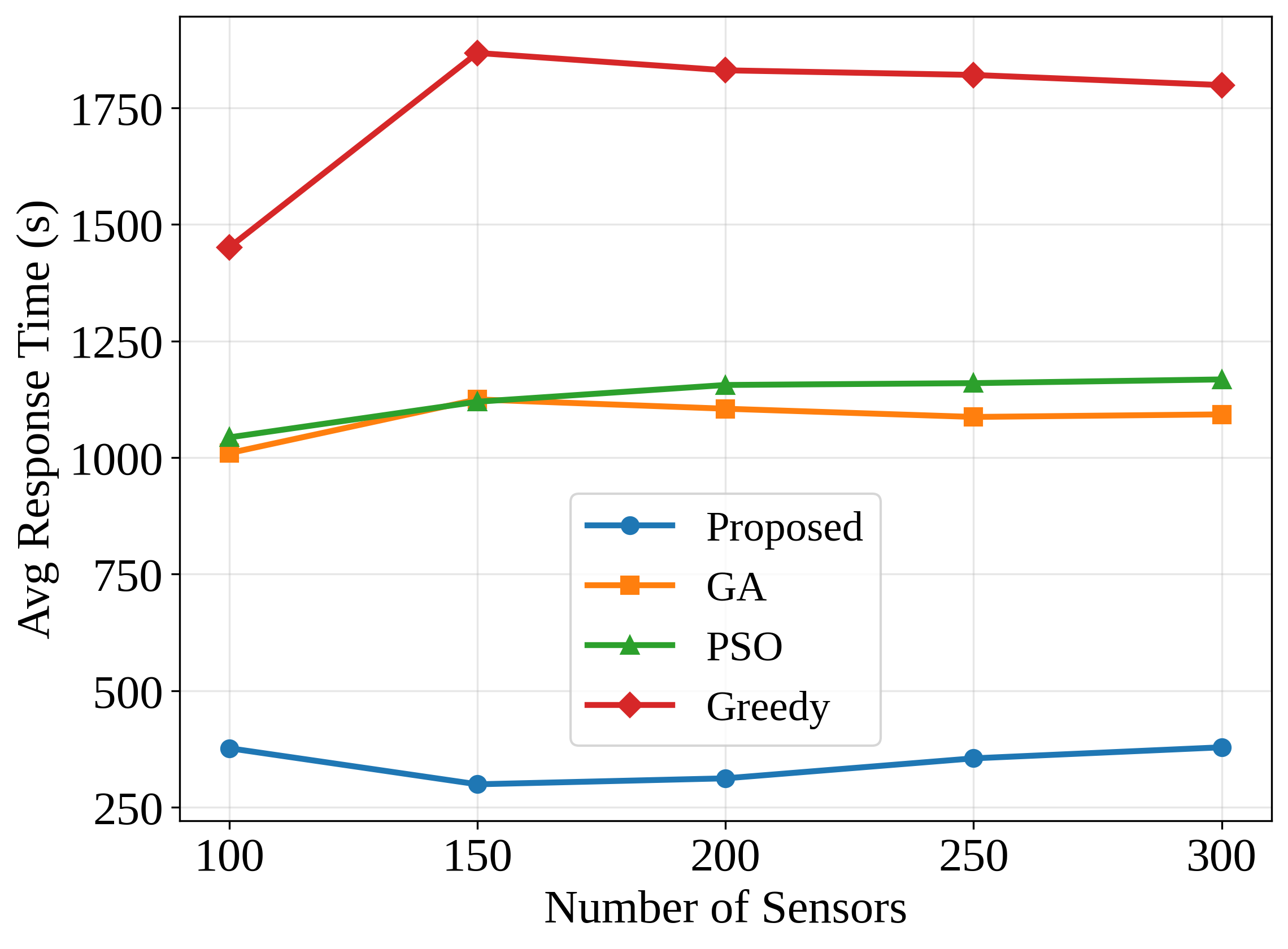}
\caption{Response time vs. scale}
\label{fig:scalability_time}
\end{subfigure}
\hfill
\begin{subfigure}[b]{0.32\textwidth}
\centering
\includegraphics[width=\textwidth]{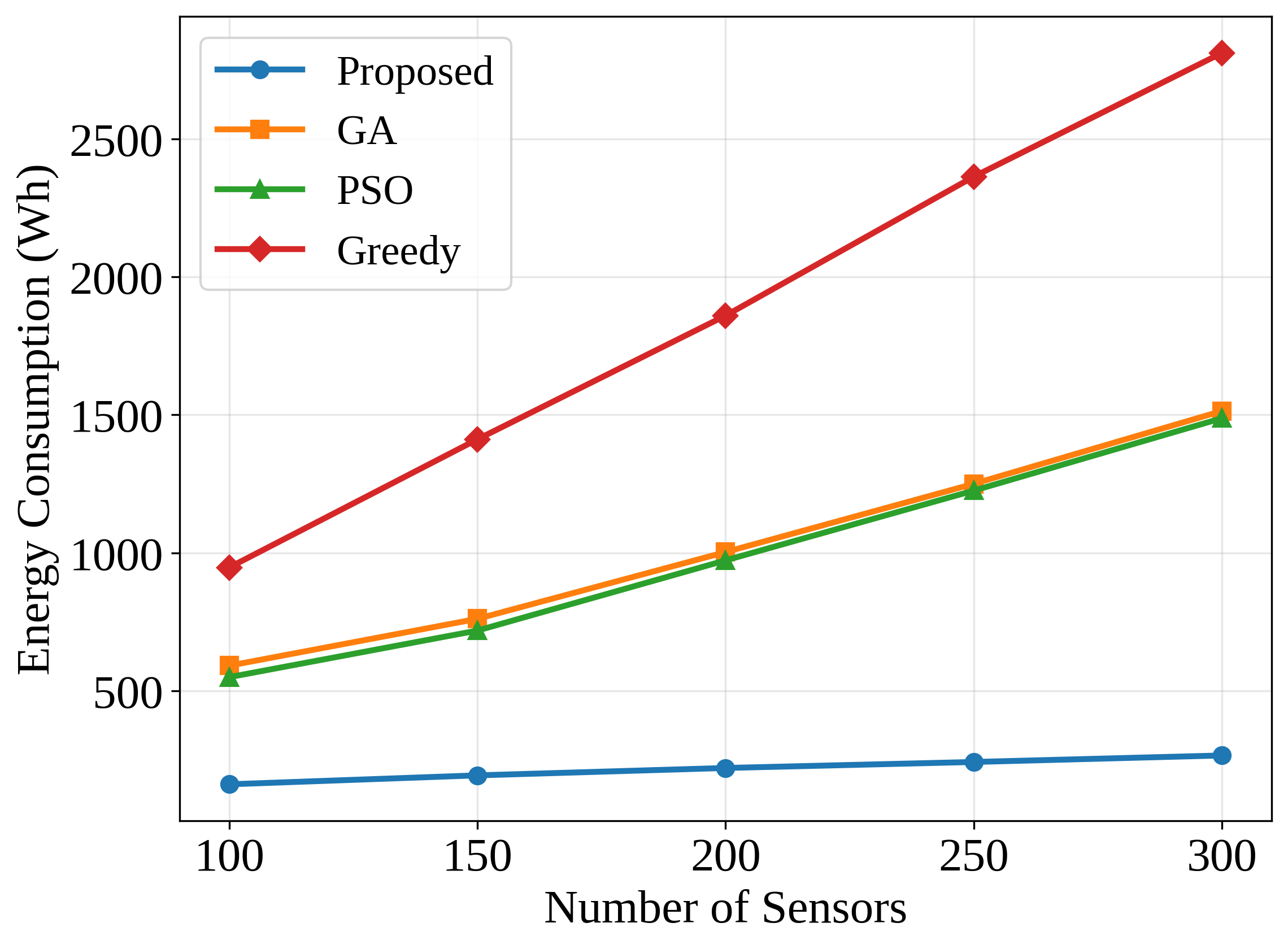}
\caption{Energy consumption vs. scale}
\label{fig:scalability_energy}
\end{subfigure}
\hfill
\begin{subfigure}[b]{0.32\textwidth}
\centering
\includegraphics[width=\textwidth]{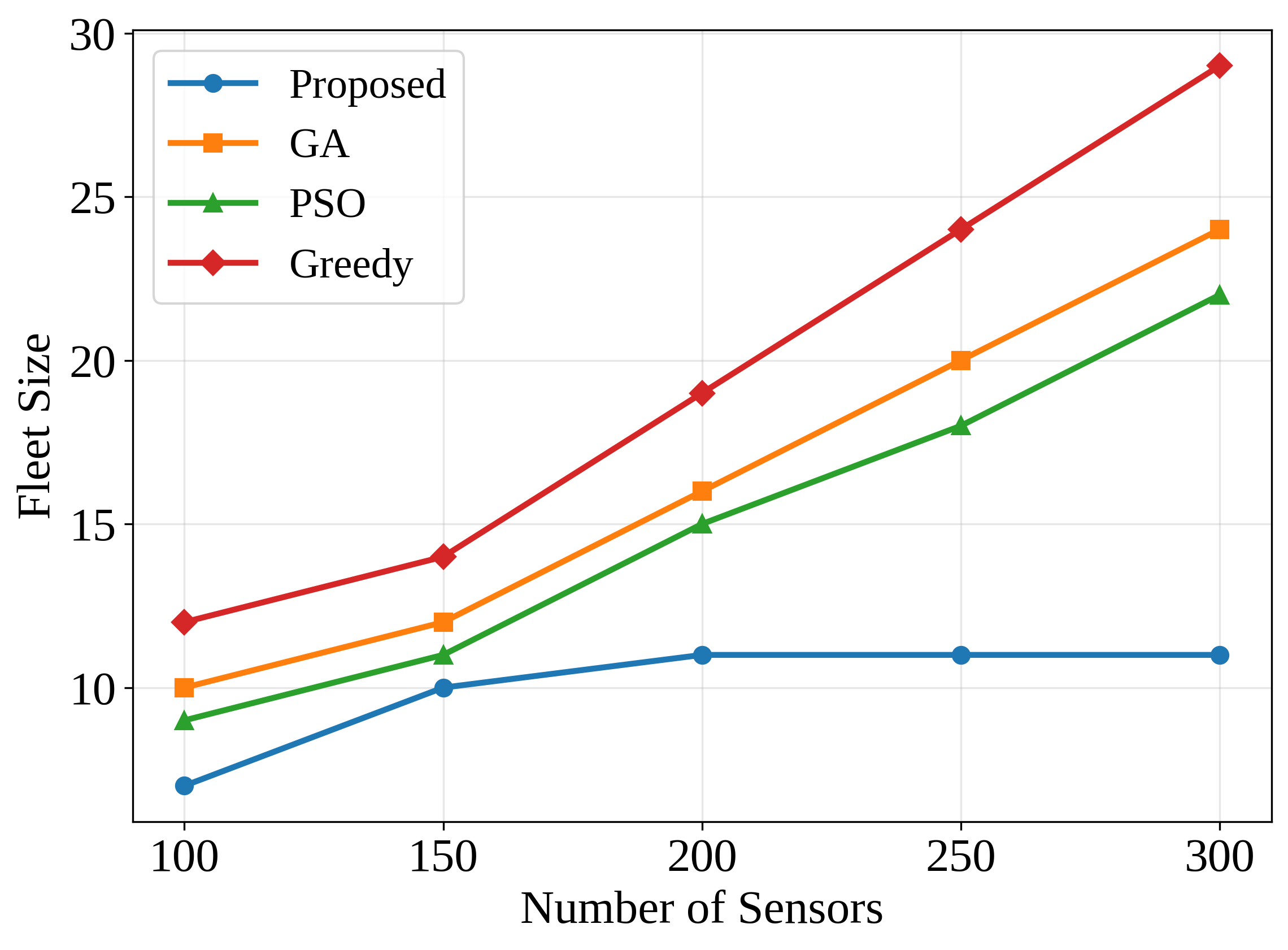}
\caption{Fleet size vs. scale}
\label{fig:scalability_fleet}
\end{subfigure}
\caption{Scalability analysis showing performance across varying sensor deployment scales (100-300 sensors).}
\label{fig:scalability}
\end{figure*}

Fig.~\ref{fig:ablation}(a) reveals an asymmetric contribution pattern. Removing 2-opt increases the average from 288 s to 339 s (17.7\% degradation), whereas removing K-means causes a jump to 1149 s (299\% degradation). This asymmetry is expected: clustering determines which sensors share a UAV, fundamentally shaping workload balance, while 2-opt only reorders visits within an already-defined cluster. The w/o Both configuration reaches 1316 s, confirming that the two components are complementary rather than redundant.

Fig.~\ref{fig:ablation}(b) reinforces this finding through the CDF. The Full configuration concentrates between 200--600 s with a steep slope, indicating uniform QoS across sensors. Removing K-means flattens the curve and extends the tail to 2000 s, revealing that random assignment creates a few severely underserved clusters. At the 95th percentile: Full 500 s, w/o 2-opt 700 s, w/o K-means 1800 s, w/o Both 1950 s.

Fig.~\ref{fig:ablation}(c) shows a consistent pattern in energy consumption. Full consumes 217 Wh; removing 2-opt increases this by 13.8\% (247 Wh), while removing K-means causes a 246\% increase (751 Wh). The disproportionate impact of clustering on energy reflects the fact that poorly balanced clusters force some UAVs into long routes spanning distant sensors.

\subsection{Emergency Service Performance}

We validate the emergency mechanism by injecting 5 emergency events at high-risk sensors (fire history $> 50$) and comparing five configurations: the proposed method with and without emergency capability, and all three baselines with emergency adaptation enabled.

Fig.~\ref{fig:emergency}(a) shows that the proposed method achieves 233 s emergency response time, well within the $T_{urgent} = 300$ s deadline. This represents a 27.8\% improvement over the same framework without emergency capability (323 s), and 50.0\%, 54.4\%, and 55.7\% improvements over GA (465 s), PSO (511 s), and Greedy (526 s) with emergency adaptation. An important observation is that adding emergency capability to baselines does not close the performance gap. GA with emergency adaptation (465 s) still exceeds the proposed method without it (323 s), demonstrating that emergency mechanisms cannot compensate for inefficient patrol planning. Shorter base routes are the prerequisite for fast emergency response.

Fig.~\ref{fig:emergency}(b) examines the cost of emergency adaptation on normal operations. Enabling the emergency mechanism has negligible impact on normal response time (315 s vs.\ 315 s for the proposed method). This confirms that the waypoint return strategy (Section~\ref{sec:emergency}) effectively contains the disruption: the diverted UAV resumes its route at the nearest point, and the brief absence does not meaningfully increase revisit times for the remaining sensors in its cluster.

\subsection{Scalability Analysis}

We evaluate scalability by varying the sensor count from 100 to 300 while keeping the monitoring area at 100 km$^2$.

Fig.~\ref{fig:scalability}(a) shows that the proposed method maintains stable response times (250--400 s, 27\% variation) as sensor density triples. The near-constant performance results from the adaptive fleet sizing mechanism: adding sensors triggers additional UAVs (Fig.~\ref{fig:scalability}(c)), keeping per-UAV cluster sizes manageable. In contrast, GA and PSO degrade to 1000--1150 s, and Greedy to 1450--1850 s.

Fig.~\ref{fig:scalability}(b) shows near-linear energy growth from 160 to 270 Wh, maintaining 73--88\% savings across all scales. At 300 sensors, the proposed method consumes 266 Wh versus 1514 Wh (GA), 1486 Wh (PSO), and 2751 Wh (Greedy).

Fig.~\ref{fig:scalability}(c) reveals the most notable scaling property. The proposed method grows from 7 to 11 UAVs (57\% increase) when sensors triple, while baselines require 140--144\% increases. This sub-linear fleet growth stems from the clustering algorithm's ability to absorb additional sensors into existing clusters when spatial density allows, only spawning new clusters when constraints are violated.

Regarding planning time, the proposed method completes in 0.1/0.3/0.8\,s for 100/200/300 sensors, compared to 5/15/35\,s (GA) and 3/10/25\,s (PSO), representing a 20--40$\times$ speedup while maintaining superior solution quality.

\section{CONCLUSIONS AND FUTURE WORK}

We have proposed an integrated framework for UAV-assisted wildfire monitoring with edge computing that jointly optimizes sensor clustering, edge assignment, patrol routing, and fleet sizing under energy, revisit time, and capacity constraints. The core design principle is that these subproblems are interdependent and must be solved together: clustering shapes both route efficiency and edge workloads, while capacity constraints feed back into feasible cluster configurations.

Experiments across multiple scales demonstrate consistent improvements over GA, PSO, and Greedy baselines: 70.6--84.2\% lower response time, 73.8--88.4\% less energy, and 26.7--42.1\% fewer UAVs. The emergency mechanism achieves 233 s response time, within the 300 s deadline and 50--56\% faster than baselines, with negligible impact on normal operations. The ablation study confirms that weighted clustering contributes the most to performance, while 2-opt provides complementary route-level refinement.

As part of our future work, we will investigate the following directions. First, we will integrate the framework with operational bushfire management systems in Australia, including Victoria's Country Fire Authority infrastructure~\cite{CFA_BMP} and the Northern Territory's Bushfire Emergency Management System~\cite{NT_BEMS}, to validate its effectiveness under real-world deployment conditions. Second, we will extend the model to support heterogeneous UAV fleets with varying speed and battery specifications, and incorporate real-time weather data such as wind speed and direction into adaptive patrol strategies that dynamically adjust routes during monitoring operations. Third, we will incorporate multi-modal sensing, including UAV-mounted thermal cameras and satellite imagery, to enable earlier fire detection and enrich the data processed by edge nodes within the same service architecture.

\bibliographystyle{IEEEtran}
\bibliography{Bibliograph_new}

\end{document}